\documentclass[fleqn,11pt,twoside]{article}

\usepackage{amsthm,amsthm,amssymb, color, xcolor,epsfig, graphics, subfigure}

\usepackage{amsmath, graphicx, latexsym, lscape }

\usepackage[all]{xy}
\usepackage{hyperref}

\makeatletter
\newcommand{\copyrightnote}[2]{{\renewcommand{\thefootnote}{}
 \footnotetext{\small\it
\begin{flushleft}
 \copyright \ #1   #2  
\end{flushleft}}}}

\newcommand{\Name}[1]{\begin{flushleft}
                       \LARGE \bf #1
                       \end{flushleft}\vspace{-3mm}}

\newcommand{\Author}[1]{\begin{flushleft}
                       \it #1 \end{flushleft}}

\newcommand{\Address}[1]{\begin{flushleft}
                       \it #1 \end{flushleft}}

\newcommand{\Date}[1]{\begin{flushleft}
                      \small  \it #1 \end{flushleft}}

\newcommand{\evenhead}{Author \ name}
\newcommand{\oddhead}{Article \ name}

\renewcommand{\@evenhead}{
\hspace*{-3pt}\raisebox{-15pt}[\headheight][0pt]{\vbox{\hbox to \textwidth
{\thepage \hfil \evenhead}\vskip4pt \hrule}}}
\renewcommand{\@oddhead}{
\hspace*{-3pt}\raisebox{-15pt}[\headheight][0pt]{\vbox{\hbox to \textwidth
{\oddhead \hfil \thepage}\vskip4pt\hrule}}}
\renewcommand{\@evenfoot}{}
\renewcommand{\@oddfoot}{}

\long\def\@makecaption#1#2{%
  \vskip\abovecaptionskip
  \sbox\@tempboxa{\small \textbf{#1.}\ \ #2}%
  \ifdim \wd\@tempboxa >\hsize
    {\small \textbf{#1.}\ \ #2}\par
  \else
    \global \@minipagefalse
    \hb@xt@\hsize{\hfil\box\@tempboxa\hfil}%
  \fi
  \vskip\belowcaptionskip}

\newcommand{\JNMPnumberwithin}[3][\arabic]{%
  \@ifundefined{c@#2}{\@nocounterr{#2}}{%
    \@ifundefined{c@#3}{\@nocnterr{#3}}{%
      \@addtoreset{#2}{#3}%
      \@xp\xdef\csname the#2\endcsname{%
        \@xp\@nx\csname the#3\endcsname .\@nx#1{#2}}}}%
}

\renewenvironment{proof}[1][\proofname]{\par
  \normalfont
  \topsep6\p@\@plus6\p@ \trivlist
  \item[\hskip\labelsep\textbf{%
    #1\@addpunct{.}}]\ignorespaces
}{%
  \qed\endtrivlist
}

\newcommand{\resetfootnoterule} {
  \renewcommand\footnoterule{%
  \kern-3\p@
  \hrule\@width.4\columnwidth
  \kern2.6\p@}
}

\renewcommand{\footnoterule}{}

\makeatother

\theoremstyle{definition}
\newtheorem{theorem}{Theorem}[section]
\newtheorem{proposition}[theorem]{Proposition}
\newtheorem{lemma}[theorem]{Lemma}
\newtheorem{corollary}[theorem]{Corollary}

\newtheorem{definition}[theorem]{Definition}
\newtheorem{example}[theorem]{Example}
\newtheorem{remark}[theorem]{Remark}
\newtheorem{notation}[theorem]{Notation}
\newtheorem{assumption}[theorem]{Assumption}

\makeatletter
\@addtoreset{equation}{section}
\makeatother

\newcommand{\newstep}[1]{
	\refstepcounter{#1}
	\paragraph{\textbf{Step \arabic{#1}.}}
}

\begin{document}

\renewcommand{\evenhead}{ {\LARGE\textcolor{blue!10!black!40!green}{{\sf \ \ \ ]ocnmp[}}}\strut\hfill T Mase}
\renewcommand{\oddhead}{ {\LARGE\textcolor{blue!10!black!40!green}{{\sf ]ocnmp[}}}\ \ \ \ \   Domain independence of the Laurent property, and more}

\thispagestyle{empty}
\newcommand{\FistPageHead}[3]{
\begin{flushleft}
\raisebox{8mm}[0pt][0pt]
{\footnotesize \sf
\parbox{150mm}{{Open Communications in Nonlinear Mathematical Physics}\ \  \ {\LARGE\textcolor{blue!10!black!40!green}{]ocnmp[}}
\ \ Vol.4 (2024) pp
#2\hfill {\sc #3}}}\vspace{-13mm}
\end{flushleft}}

\FistPageHead{1}{\pageref{firstpage}--\pageref{lastpage}}{ \ \ Article}

\strut\hfill

\strut\hfill

\copyrightnote{The author(s). Distributed under a Creative Commons Attribution 4.0 International License}

\Name{A study on the domain independence of the Laurent property, the irreducibility and the coprimeness in lattice equations}

\Author{Takafumi Mase}

\Address{Graduate School of Mathematical Sciences, 
the University of Tokyo, 3-8-1 Komaba, Meguro-ku, Tokyo 153-8914, Japan.}

\Date{Received July 20, 2024; Accepted August 2, 2024}

\setcounter{equation}{0}

\begin{abstract}
\noindent 
We study the Laurent property, the irreducibility and the coprimeness for lattice equations (partial difference equations), mainly focusing on how the choice of initial value problem (the choice of domain) affects these properties.
We show that these properties do not depend on the choice of domain as long as they are considered together.
In other words, these properties are inherent to a difference equation.
Applying our result, we discuss the reductions of lattice equations.
We show that any reduction of a Laurent system, even if the lattices have torsion elements, preserves the Laurent property.
\end{abstract}

\label{firstpage}


\section{Introduction}\label{section:introduction}

A difference equation is said to have the Laurent property if every iterate is expressed as a Laurent polynomial in the initial values \cite{Laurent_phenomenon}.
Such a difference equation is often called a Laurent system \cite{Laurent_recurrences, QRT_Laurent, R_systems}.

This property, also known as the Laurent phenomenon, is closely related to the theory of cluster algebras.
It is well known that each cluster variable is expressed as a Laurent polynomial in the initial cluster \cite{cluster1}.
Therefore, one can generate a Laurent system using a specific sequence of mutations for a cluster algebra \cite{Fordy_Marsh, Fordy_Hone}.
LP (Laurent phenomenon) algebras were proposed to generalize cluster algebras \cite{LP_algebra}.
The Laurent property for cluster algebras and LP algebras is shown using the Caterpillar Lemma.

A basic strategy to prove the Laurent property is to show that several iterates are mutually coprime as Laurent polynomials.
For example, this idea is used in Hickerson's method \cite{Gale, Robinson} and the proof of the Caterpillar Lemma \cite{Laurent_phenomenon}.
On the other hand, the coprimeness for a Laurent system is a global property and focuses on the coprimeness of every pair of iterates \cite{coprimeness}.
We will review the definition of the coprimeness for Laurent systems in Definition~\ref{definition:irreducibility_coprimeness}.

We say that a Laurent system satisfies the irreducibility if every iterate is irreducible (or a unit) as a Laurent polynomial in the initial variables (Definition~\ref{definition:irreducibility_coprimeness}).
This property holds for many Laurent systems and often helps us show the Laurent property and the coprimeness.
However, some Laurent systems are known to satisfy the coprimeness but do not have the irreducibility \cite{factorize}.

Some non-Laurent systems can be transformed into Laurent systems by changing dependent variables, and this procedure is called Laurentification \cite{Laurent_superintegrable, coprimeness, QRT_Laurent}.
The transformation of dependent variables corresponds to the singularity patterns of the original equation.
Therefore, Laurentification, in conjunction with the coprimeness, reveals the (global) singularity structure for non-Laurent systems, leading to the concept of the coprimeness for non-Laurent systems \cite{coprimeness}.
In addition, this procedure allows us to compute the degree growth for the original equation \cite{Laurent_superintegrable, extended_hv, lattice_degree}.
When calculating the degree growth, we often use the tropical dynamics of the Laurentified equation.

When considering the Laurent property for lattice equations (partial difference equations), the choice of initial value problem, i.e., the choice of domain, is essential.
Since an initial value problem for a lattice equation has infinitely many initial variables, some conditions on a domain are necessary when we show the Laurent property.
In fact, a Laurent system often loses the Laurent property if considered on a pathological domain \cite{rims}.
The conditions on a domain required for the Laurent property were first considered for several concrete lattice equations in \cite{Laurent_phenomenon} and then clearly stated for discrete Hirota bilinear equations in \cite{rims,investigation}.
The conditions for general lattice equations, defined in Definition~\ref{definition:lightcone_regular}, are described by the future and past light cones \cite{domain}.

Some Laurent systems on multi-dimensional lattices can be obtained from a (sequence of) mutation(s) for a cluster or LP algebra.
For example, Okubo constructed quivers that generate the discrete KdV (Korteweg-de Vries) equation and the Hirota-Miwa equation (the bilinear form of the discrete KP equation), respectively \cite{Okubo_rims,Okubo_HM}, and later obtained an LP-algebraic expression of the bilinear form of the discrete BKP equation \cite{Okubo_BKP}.
Even some non-integrable Laurent systems on multi-dimensional lattices are constructed from LP algebras \cite{toda50}.
Note that a cluster (or LP) algebra that generates a lattice equation defined by a single recurrence relation is not of finite rank since an initial value problem has infinitely many initial variables.
On the other hand, the author and collaborators recently constructed a quiver of finite rank whose commuting sequences of mutations generate a lattice equation on $\mathbb{Z}^2$ \cite[$A_3$-case]{Hone_Kim_Mase}.

Proving the Laurent property for a lattice equation without constructing a cluster or LP algebra, i.e., without relying on the Caterpillar Lemma, is sometimes very difficult because no general method to show that several iterates are mutually coprime is known.
In this process, we need to check that an iterate does not divide another iterate.
For now, the most commonly used strategy to show this is to check that an iterate is irreducible by substituting numerical values for (some of) the initial variables or by focusing on a specific initial variable that appears only in one of two iterates.
This strategy is most effective when working on a specific domain, such as a translation-invariant one \cite{extoda, toda50, exkdv}.

In this paper, we consider how the choice of domain affects the Laurent property, the irreducibility and the coprimeness of an autonomous, invertible difference equation defined by a single recurrence relation.
We will prove that as long as considered together, these three properties do not depend on the choice of domain (Theorem~\ref{theorem:main}).


Let us fix notations and definitions and see what conditions we require on equations.

\begin{notation}\label{notation_lattice}
	Let $L$ be a lattice on which we define an equation.
	In this paper, we use the term ``lattice'' to mean a finitely generated $\mathbb{Z}$-module.

	Let $R$ be a UFD (unique factorization domain), such as $\mathbb{Z}$, $\mathbb{Z}[a_1, a_2, \cdots]$ and a field.
	We consider the Laurent property over the coefficient ring $R$.
	Let $K$ be the field of fractions of $R$ and we use $\overline{K}$ to denote its algebraic closure.

	We say that a nonzero element $f$ of a UFD is irreducible if $f$ does not factorize into a product of two non-unit elements.
	It should be noted that, according to this definition, units are also labeled ``irreducible,'' while in the field of algebra, units are not considered irreducible.
\end{notation}

\begin{remark}\label{remark:basis}
	In our terminology, $L$ can possess a torsion element other than $0$.
	Some Laurent systems with multiple tau functions, such as the bilinear form of the discrete mKdV (modified Korteweg-de Vries) equation (Example~\ref{example:discrete_mkdv}), can be thought of as defined over a lattice with torsion elements.
	By the fundamental theorem of finitely generated abelian groups, taking appropriate generators, we can express $L$ as
	\[
		L \cong \mathbb{Z}^{\operatorname{rank} L} \oplus (\mathbb{Z} / a_1 \mathbb{Z}) \oplus \cdots \oplus (\mathbb{Z} / a_m \mathbb{Z})
	\]
	for some $a_1, \ldots, a_m \in \mathbb{Z}_{> 1}$.
	However, we do not use this expression in this paper.
\end{remark}

\begin{assumption}\label{assumption:general_equation}
	In this paper, we will consider equations of the following form:
	\begin{equation}\label{equation:general_equation}
		f_h = \Phi(f_{h + v_1}, \ldots, f_{h + v_N})
		\quad
		(h \in L),
	\end{equation}
	where $\Phi$ is an $R$-coefficient irreducible Laurent polynomial with $N$-variables and $v_1, \ldots, v_N \in L$ are shift vectors.
	We suppose that $v_1, \ldots, v_N$ are all different and $\Phi$ does depend on all of $f_{h + v_1}, \ldots, f_{h + v_N}$.
	If not, we decrease the number of variables of $\Phi$.

	If $\Phi$ is a Laurent monomial, then the equation defined by $\Phi$ can be linearized using a logarithm.
	Therefore, we may assume that $\Phi$ is not a Laurent monomial.
	In our case, $\Phi$ is primitive over $R$ (see Definition~\ref{definition:primitive}).

	Moreover, we may assume that the shift vectors $v_1, \ldots, v_N$ generate $L$ as a lattice.
	If not, we replace $L$ with $\operatorname{span}_{\mathbb{Z}} (v_1, \ldots, v_N)$.
\end{assumption}

\begin{remark}
	We do not explicitly exclude the cases where the numerator of $\Phi$ has a monomial factor.
	However, as studied in \cite{sc_Laurent} in the ordinary difference case, such a system does not satisfy the irreducibility or the coprimeness and naturally falls outside the scope of our research.
\end{remark}

\begin{assumption}[\cite{investigation}]\label{assumption:linear_independence}
	We think of equation \eqref{equation:general_equation} as defining an evolution on $L$ in the direction from $f_{h + v_1}, \ldots, f_{h + v_N}$ to $f_h$.
	Therefore, we assume that the shift vectors $v_1, \ldots, v_N$ are linearly independent over $\mathbb{Z}_{\ge 0}$.
	That is, a relation
	\[
		\sum_i a_i v_i = 0
	\]
	for $a_1, \ldots, a_N \in \mathbb{Z}_{\ge 0}$ implies
	\[
		a_1 = \cdots = a_N = 0.
	\]
\end{assumption}

\begin{definition}[\cite{rims,investigation}]
	Define an additive monoid $S \subset L$ as
	\[
		S = \operatorname{span}_{\mathbb{Z}_{\ge 0}} (v_1, \ldots, v_N)
		= \left\{ \sum_i a_i v_i \mid a_1, \ldots, a_N \in \mathbb{Z}_{\ge 0} \right\}
	\]
	and let ``$\le$'' be the binary relation on $L$ defined as
	\[
		h_1 \le h_2 \quad \Leftrightarrow \quad h_1 - h_2 \in S.
	\]
	That is, $h_1 \le h_2$ holds if and only if there exist $a_1, \ldots, a_N \in \mathbb{Z}_{\ge 0}$ such that
	\[
		h_1 = h_2 + \sum_i a_i v_i.
	\]
	By Assumption~\ref{assumption:linear_independence}, this binary relation is a (partial) order on $L$.
	Clearly, $\le$ is invariant under translations on $L$, i.e., for $h_1 \le h_2$ and $v \in L$, we have $h_1 + v \le h_2 + v$.
\end{definition}

\begin{assumption}\label{assumption:invertible}
	We assume that equation \eqref{equation:general_equation} is invertible in the following sense.

	First, we assume that the set of shift vectors
	\[
		\{ 0, v_1, \ldots, v_N \}
	\]
	has a unique minimal element (i.e., a minimum element) with respect to the order $\le$.
	The minimum element is not $0$ since it is the maximum.
	As the ordering of $v_1, \ldots, v_N$ can be changed at will, we may assume that the minimum element is $v_N$.

	Solving \eqref{equation:general_equation} in the opposite direction is to solve it with respect to $f_{h + v_N}$.
	Here, we suppose that \eqref{equation:general_equation} can be solved with respect to $f_{h + v_N}$ by a Laurent polynomial as
	\[
		f_{h + v_N} = \Psi(f_{h + v_1}, \ldots, f_{h + v_{N - 1}}, f_h),
	\]
	where $\Psi$ is an $R$-coefficient Laurent polynomial.
	Since $v_N$ is the minimum element of the set $\{ 0, v_1, \ldots, v_N \}$, the above equation defines an evolution in the opposite direction.

	Finally, we assume that $\Psi$ is irreducible as an $R$-coefficient Laurent polynomial.
\end{assumption}

\begin{remark}[\cite{rims,investigation}]
	All discrete Hirota bilinear equations satisfy Assumption~\ref{assumption:invertible} because of the balancing condition of the shift vectors.
\end{remark}

\begin{remark}
	$\Psi$ is not a Laurent monomial since $\Phi$ is not a Laurent monomial.
\end{remark}

\begin{remark}
	Although we assume that equation \eqref{equation:general_equation} can be solved in the opposite direction by an irreducible Laurent polynomial, we do not require the Laurent property in the opposite direction.
	Thus, our main theorem is valid even for equations that have the Laurent property only in one direction.
\end{remark}

\begin{definition}[light-cone regular domain \cite{Laurent_phenomenon,rims,investigation,domain}]\label{definition:lightcone_regular}
	A nonempty subset $H \subset L$ is said to be a \emph{light-cone regular} domain for equation \eqref{equation:general_equation} if it satisfies the following conditions:
	\begin{enumerate}
		\item\label{enumerate:lightcone_regular_1}
		For any $h \in H$, the intersection between the domain and the past light cone emanating from $h$ is a finite set, i.e., the set
		\[
			\{ h' \in H \mid h' \le h \} \quad (= H \cap (h + S))
		\]
		is finite.

		\item\label{enumerate:lightcone_regular_2}
		For any $h \in H$, the future light cone emanating from $h$ is contained in the domain, i.e.,
		\[
			\{ h' \in L \mid h' \ge h \} \subset H \quad
			(\Leftrightarrow h - S \subset H).
		\]

	\end{enumerate}
	For a light-cone regular domain $H$, we define $H_0 \subset H$ as
	\[
		H_0 = \{ h \in H \mid \exists i \colon h + v_i \notin H \},
	\]
	which we call the \emph{initial boundary} for $H$.
\end{definition}

\begin{remark}
	A domain that satisfies the above conditions is sometimes called a ``good domain,'' which was named by the present author \cite{rims}.
	In this paper, however, we introduce the new term ``light-cone regular,'' which is more appropriate because the name represents its properties.
\end{remark}

\begin{remark}
	The terms ``past and future light cones'' for difference equations have first been explicitly defined in \cite{domain} for quad equations.
	In the case of equation \eqref{equation:general_equation}, using the monoid $S$, the past (resp.\ future) light-cone emanating from $h \in L$ can be easily expressed as $h + S$ (resp.\ $h - S$).
\end{remark}

\begin{notation}[$R_T$, $K_T$]
	For a subset $T \subset L$, we introduce the notation
	\[
		R_T = R[f^{\pm 1}_h \mid h \in T], \quad
		K_T = K[f^{\pm 1}_h \mid h \in T].
	\]
\end{notation}

It is obvious that $R_T \subset R_{T'}$ for $T \subset T'$.
Note that $f_h$ ($h \in T$) are not always independent because they can have relations determined by equation \eqref{equation:general_equation}.
We will often use $R_{H_0}$ and $K_{H_0}$ in this paper.
In such cases, $f_h$ ($h \in H_0$) do not have any relation, i.e., they are algebraically independent over $K$, and thus can be thought of as if they were independent variables.

\begin{definition}[Laurent property]
	Let $H \subset L$ be light-cone regular.
	Equation \eqref{equation:general_equation} has the \emph{Laurent property} on $H$ if every iterate can be written as a Laurent polynomial in the initial variables, i.e.,
	\[
		f_h \in R_{H_0} \quad
		(= R[ f^{\pm 1}_{h_0} \mid h_0 \in H_0 ])
	\]
	for all $h \in H$.
\end{definition}

\begin{definition}[irreducibility, coprimeness]\label{definition:irreducibility_coprimeness}
	Let $H \subset L$ be light-cone regular and suppose that equation \eqref{equation:general_equation} has the Laurent property on $H$.
	\begin{enumerate}
		\item 
		Equation \eqref{equation:general_equation} has the \emph{irreducibility} if every iterate is irreducible as a Laurent polynomial, i.e., for any $h \in H$, $f_h$ is an irreducible element of the ring $R_{H_0}$.

		\item
		Equation \eqref{equation:general_equation} has the \emph{coprimeness} if every pair of iterates is coprime, i.e., for any $h_1, h_2 \in H$ ($h_1 \ne h_2$), $f_{h_1}$ and $f_{h_2}$ are coprime to each other in the ring $R_{H_0}$.

	\end{enumerate}
\end{definition}

As we will see in Example~\ref{example:cluster_1d}, a Laurent system sometimes has a unit element as an iterate.
The following proposition, which we will prove in Section~\ref{section:proof}, guarantees that such a phenomenon can never occur if a Laurent system has the coprimeness.

\begin{proposition}\label{proposition:not_unit}
	Suppose that equation \eqref{equation:general_equation} has the Laurent property and the coprimeness on a light-cone regular domain $H$.
	Then, for $h \in H$, $f_h$ is a unit in $R_{H_0}$ if and only if $h \in H_0$.
\end{proposition}

The following is our main theorem in this paper.

\begin{theorem}\label{theorem:main}
	Let $\widetilde{H} \subset L$ be light-cone regular and suppose that equation \eqref{equation:general_equation} has the Laurent property, the irreducibility and the coprimeness on $\widetilde{H}$.
	Then, for any light-cone regular domain $H \subset L$, equation \eqref{equation:general_equation} has these three properties on $H$, too.
	That is, if considered together, these three properties do not depend on the choice of light-cone regular domain.
\end{theorem}

Thanks to this theorem, once we prove that an equation has these three properties on one domain, these properties immediately follow on any light-cone regular domain.

The paper is organized as follows.
Section~\ref{section:examples} consists of examples of Laurent systems.
We see that most Laurent systems in the literature satisfy the assumptions and conditions introduced in the introduction.
We prove our main theorem in Section~\ref{section:proof}.
Our main strategy is to substitute specific numerical values for initial variables, even in the case of general lattice equations.
In Section~\ref{section:reduction}, we consider reductions of Laurent systems, i.e., the procedure to obtain a difference equation defined on a lower-dimensional lattice by requiring some translation invariance on the original equation.
Reductions are the most important application of showing the Laurent property on a general domain.
We show in Proposition~\ref{proposition:reduction} that any reduction of a Laurent system is also a Laurent system, even if the lattices have torsion elements.
In Appendix~\ref{appendix:algebra}, we recall some basic properties of Laurent polynomial rings, which play an essential role in Section~\ref{section:proof}.

\section{Examples}\label{section:examples}

In Section~\ref{section:introduction}, we introduced some assumptions on equation \eqref{equation:general_equation}.
In this section, we check that many Laurent systems in the literature satisfy these conditions.

\begin{example}[discrete KdV equation]\label{example:discrete_kdv}
	The bilinear form of the discrete KdV equation is
	\[
		f_{m, n} = \frac{\alpha f_{m - 2, n} f_{m, n - 1} + \beta f_{m - 1, n} f_{m - 1, n - 1}}{f_{m - 2, n - 1}},
	\]
	where $\alpha$ and $\beta$ are nonzero constants \cite{Hirota_dkdv}.
	Let $R$ be an arbitrary UFD and let $\alpha, \beta \in R \setminus \{ 0 \}$ be coprime to each other in $R$.
	We sometimes impose a condition on $\alpha, \beta$ such as $\alpha + \beta = 1$, but in this paper, we do not need such a restriction.
	Let $L = \mathbb{Z}^2$ and define $v_1, \ldots, v_5 \in L$ as
	\begin{itemize}
		\item
		$v_1 = ( - 2, 0 )$,
		
		\item
		$v_2 = ( 0, - 1 )$,
		
		\item
		$v_3 = ( - 1, 0 )$,
		
		\item
		$v_4 = ( - 1, - 1 )$,
		
		\item
		$v_5 = ( - 2, - 1 )$.

	\end{itemize}
	Using $h = (m, n)$, the discrete KdV equation is written as
	\[
		f_h = \frac{\alpha f_{h + v_1} f_{h + v_2} + \beta f_{h + v_3} f_{h + v_4}}{ f_{h + v_5}}.
	\]
	The graphical representation of the shift vectors $0, v_1, \ldots, v_5$ is
	\[
		\begin{matrix}
			v_1 & v_3 & 0 \\
			v_5 & v_4 & v_2
		\end{matrix},
	\]
	which leads to the linear independence of $v_1, \ldots, v_5$ over $\mathbb{Z}_{\ge 0}$.
	The monoid $S$ generated by these shift vectors is
	\[
		S = \{ (m, n) \in L \mid m, n \le 0 \}
	\]
	and the corresponding order $\le$ on $L$ is expressed as
	\[
		(m, n) \le (m', n')
		\quad \Leftrightarrow \quad
		m \le m' \text{ and } n \le n'.
	\]
	With respect to this order, $v_5$ is the minimum element of the set $\{ 0, v_1, \ldots, v_5 \}$.
	Solving the discrete KdV equation in the opposite direction, we obtain
	\[
		f_{h + v_5} = \frac{\alpha f_{h + v_1} f_{h + v_2} + \beta f_{h + v_3} f_{h + v_4}}{f_{h}},
	\]
	the RHS of which is an irreducible Laurent polynomial.
	Therefore, the discrete KdV equation satisfies all the assumptions in Section~\ref{section:introduction}.
	It is well known that the discrete KdV equation has the Laurent property, the irreducibility and the coprimeness on any light-cone regular domain \cite{Laurent_phenomenon,rims,coprimeness}.
\end{example}

\begin{example}[discrete mKdV equation]\label{example:discrete_mkdv}
	The bilinear form of the discrete mKdV equation is written by two functions $f^{(0)}$ and $f^{(1)}$ as
	\begin{align*}
		f^{(0)}_{\ell, m} f^{(1)}_{\ell - 1, m - 1} &= \alpha f^{(0)}_{\ell - 1, m} f^{(1)}_{\ell, m - 1} + \beta f^{(0)}_{\ell, m - 1} f^{(1)}_{\ell - 1, m}, \\
		f^{(1)}_{\ell, m} f^{(0)}_{\ell - 1, m - 1} &= \alpha f^{(1)}_{\ell - 1, m} f^{(0)}_{\ell, m - 1} + \beta f^{(1)}_{\ell, m - 1} f^{(0)}_{\ell - 1, m},
	\end{align*}
	where $\alpha$ and $\beta$ are nonzero constants \cite{Hirota_dmkdv}.
	Let $R$ be an arbitrary UFD and let $\alpha, \beta \in R \setminus \{ 0 \}$ be coprime to each other in $R$.
	Let $L = \mathbb{Z}^2 \oplus (\mathbb{Z} / 2 \mathbb{Z})$ and define $v_1, \ldots, v_5 \in L$ as
	\begin{itemize}
		\item
		$v_1 = (- 1, 0, 0)$,

		\item
		$v_2 = (0, - 1, 1)$,

		\item
		$v_3 = (0, - 1, 0)$,

		\item
		$v_4 = (- 1, 0, 1)$,

		\item
		$v_5 = (- 1, - 1, 1)$,

	\end{itemize}
	where the third entries correspond to the $\mathbb{Z} / (2 \mathbb{Z})$-part.
	Then, the discrete mKdV equation is written by a single function $f_h$ ($h \in L$) as
	\[
		f_h = \frac{\alpha f_{h + v_1} f_{h + v_2} + \beta f_{h + v_3} f_{h + v_4}}{f_{h + v_5}},
	\]
	the RHS of which is an irreducible Laurent polynomial.
	Since the $\ell$- and $m$-coordinates of each generator are non-positive but do not vanish simultaneously, these shift vectors are linearly independent over $\mathbb{Z}_{\ge 0}$.
	The monoid $S$ generated by the shift vectors is
	\[
		S = \{ (\ell, m, n) \in L \mid \ell, m \le 0 \}
	\]
	and the corresponding order $\le$ on $L$ is expressed as
	\[
		(\ell, m, n) \le (\ell', m', n')
		\quad \Leftrightarrow \quad
		\ell \le \ell' \text{ and } m \le m'.
	\]
	The shift $v_5$ is the minimum element of the set $\{ 0, v_1, \ldots, v_5 \}$ and we can solve the equation for $f_{h + v_5}$ by an irreducible Laurent polynomial.
	Therefore, this equation satisfies all the assumptions in Section~\ref{section:introduction}.
\end{example}

\begin{remark}
	Not only the discrete (modified) KdV equation but all the discrete Hirota bilinear equations, including the Hirota-Miwa equation, the bilinear form of the discrete BKP equation and their nontrivial reductions, satisfy the assumptions in Section~\ref{section:introduction}.
\end{remark}

Although our main topic in this paper is lattice equations, we do not exclude Laurent systems defined on a one-dimensional lattice.
However, as the following two examples show, some Laurent systems on a one-dimensional lattice do not have the irreducibility or the coprimeness.

\begin{example}\label{example:cluster_1d}
	Let $r \in \mathbb{Z}_{\ge 1}$ and consider the system
	\[
		f_n = \frac{f^r_{n - 1} + 1}{f_{n - 2}}
	\]
	over the coefficient ring $\mathbb{Z}$.
	It is well-known that this equation can be obtained from a cluster algebra and has the Laurent property, i.e., $f_n \in \mathbb{Z}[f^{\pm 1}_0, f^{\pm 1}_1]$ for all $n \ge 0$ \cite{cluster1}.

	If $r = 1$, this system is the Lyness map \cite{Lyness}.
	It is periodic with period $5$ and has the irreducibility.
	In particular, it does not have the coprimeness since $f_{n + 5} = f_n$.
	Since $f_5 = f_0$ is a unit in $\mathbb{Z}[f^{\pm 1}_0, f^{\pm 1}_1]$, this equation has a nontrivial unit as an iterate.

	On the other hand, if $r \ge 2$, the equation is not periodic anymore.
	In this case, the equation has the irreducibility over $\mathbb{Z}$ if and only if $r$ is a power of $2$ (i.e., the first iterate is irreducible).
	However, the equation always has the coprimeness \cite{factorize}.
\end{example}

\begin{example}\label{example:hirota_miwa_not_coprime}
	Consider the equation
	\begin{equation}
		f_n = \frac{f_{n - 1} f_{n - 6} + f_{n - 3} f_{n - 4}}{f_{n - 7}}
	\end{equation}
	on $H = \mathbb{Z}_{\ge 0}$, i.e., the initial values are $f_0, \ldots, f_6$.
	Since it can be obtained as a reduction of the Hirota-Miwa equation, this equation has the Laurent property.
	However, since
	\begin{align*}
		f_7 &= \frac{f_1 f_6 + f_3 f_4}{f_0}, \\
		f_8 &= \frac{f_0 f_4 f_5 + f_1 f_2 f_6 + f_2 f_3 f_4}{f_0 f_1}, \\
		f_9 &= \frac{(f_0 f_5 + f_2 f_3) (f_1 f_6 + f_3 f_4)}{f_0 f_1 f_2},
	\end{align*}
	this equation does not satisfy the irreducibility or the coprimeness.
	Using
	\[
		p = f_1 f_6 + f_3 f_4, \quad
		q = f_0 f_4 f_5 + f_1 f_2 f_6 + f_2 f_3 f_4, \quad
		r = f_0 f_5 + f_2 f_3,
	\]
	we can write the time evolution of the equation as
	\begin{align*}
		f_0 &\mapsto f_1, \quad
		f_1 \mapsto f_2, \quad
		f_2 \mapsto f_3, \\
		f_3 &\mapsto f_4, \quad
		f_4 \mapsto f_5, \quad
		f_5 \mapsto f_6, \quad
		f_6 \mapsto \frac{p}{f_0}, \\
		p &\mapsto \frac{q}{f_0}, \quad
		q \mapsto \frac{p r}{f_0}, \quad
		r \mapsto p.
	\end{align*}
	Let
	\[
		A = \begin{bmatrix}
			0 & 1 & 0 & 0 & 0 & 0 & 0 & 0 & 0 & 0 \\
			0 & 0 & 1 & 0 & 0 & 0 & 0 & 0 & 0 & 0 \\
			0 & 0 & 0 & 1 & 0 & 0 & 0 & 0 & 0 & 0 \\
			0 & 0 & 0 & 0 & 1 & 0 & 0 & 0 & 0 & 0 \\
			0 & 0 & 0 & 0 & 0 & 1 & 0 & 0 & 0 & 0 \\
			0 & 0 & 0 & 0 & 0 & 0 & 1 & 0 & 0 & 0 \\
			- 1 & 0 & 0 & 0 & 0 & 0 & 0 & 1 & 0 & 0 \\
			- 1 & 0 & 0 & 0 & 0 & 0 & 0 & 0 & 1 & 0 \\
			- 1 & 0 & 0 & 0 & 0 & 0 & 0 & 0 & 0 & 1 \\
			0 & 0 & 0 & 0 & 0 & 0 & 0 & 1 & 0 & 0 \\
		\end{bmatrix}
	\]
	and define $F^{(0)}_n, \ldots, F^{(6)}_n, P_n, Q_n, R_n$ for $n \ge 0$ by
	\[
		F^{(0)}_0 = 1, \quad
		F^{(1)}_0 = \cdots = F^{(6)}_0 = 0, \quad
		P_0 = Q_0 = R_0 = 0
	\]
	and
	\[
		\begin{bmatrix}
			F^{(0)}_n \\
			\vdots \\
			F^{(6)}_n \\
			P_n \\
			Q_n \\
			R_n
		\end{bmatrix}
		= A \begin{bmatrix}
			F^{(0)}_{n - 1} \\
			\vdots \\
			F^{(6)}_{n - 1} \\
			P_{n - 1} \\
			Q_{n - 1} \\
			R_{n - 1}
		\end{bmatrix}
		\quad
		(n \ge 1).
	\]
	Then, the general iterate of this equation is expressed as
	\[
		f_n = f^{F^{(0)}_n}_0 f^{F^{(1)}_n}_1 \cdots f^{F^{(6)}_n}_6 p^{P_n} q^{Q_n} r^{R_n},
	\]
	which implies that the exponent of each factor grows quadratically.
	This equation can be reduced to the Lyness map via the transformation $g_n = \frac{f_n f_{n - 5}}{f_{n - 2} f_{n - 3}}$ \cite{Hone_Kouloukas}.
\end{example}

When considering an equation on $\mathbb{Z}$, we almost always use the domain
\[
	\{ n \in \mathbb{Z} \mid n \ge a \}
\]
for some $a \in \mathbb{Z}$.
However, some ordinary difference equations have another type of light-cone regular domains.

\begin{example}
	Consider the discrete Hirota bilinear equation
	\[
		f_n = \frac{\alpha f_{n - 2} f_{n - 5} + \beta f_{n - 3} f_{n - 4}}{f_{n - 7}},
	\]
	where $\alpha$ and $\beta$ are nonzero constants.
	Then, the domain
	\[
		H = \{ 0 \} \cup \{ n \in \mathbb{Z} \mid n \ge 2 \}
	\]
	is light-cone regular for this equation because $f_{n - 1}$ does not appear in the RHS of the equation.
	The initial boundary for $H$ is
	\[
		H_0 = \{ 7 \} \cup \{ n \in \mathbb{Z} \mid n \ge 9 \}.
	\]
	Since this equation satisfies all the assumptions in Section~\ref{section:introduction}, we can apply Theorem~\ref{theorem:main} to this equation.
	That is, if it has the three properties on $H$, then the equation also satisfies them on $\{ n \in \mathbb{Z} \mid n \ge a \}$.
\end{example}

It is known that some non-integrable systems have the Laurent property, the irreducibility and the coprimeness.

\begin{example}[extensions of the discrete Toda equation \cite{toda50}]\label{example:toda50}
	Let $R = \mathbb{Z}$, $a, b \in \mathbb{Z}_{> 0}$ and let $k_1, \ldots, k_{a + b}, \ell_1, \ldots, \ell_{a + b} \in \mathbb{Z}_{> 0}$.
	Suppose that
	\[
		\operatorname{GCD}(k_1, \ldots, k_{a + b}, \ell_1, \ldots, \ell_{a + b})
	\]
	is a power of $2$, where $\operatorname{GCD}$ stands for the greatest common divisor.
	Express a vector $\mathbf{n} \in \mathbb{Z}^{a + b}$ as
	\[
		\mathbf{n} = (n_1, \ldots, n_{a + b})
		= \sum^{a + b}_{i=1} n_i \mathbf{e}_i,
	\]
	where $\mathbf{e}_i \in \mathbb{Z}^{a + b}$ is the unit vector in the $i$-th direction.
	Then, an extension of the discrete Toda equation is
	\begin{equation}\label{equation:toda50}
		f_{t, \mathbf{n}} = \frac{ \prod^a_{i = 1} f^{k_i}_{t - 1,\mathbf{n} + \mathbf{e}_i} f^{\ell_i}_{t - 1,\mathbf{n} - \mathbf{e}_i} + \prod^{a + b}_{j = a + 1} f^{k_j}_{t - 1,\mathbf{n} + \mathbf{e}_j} f^{\ell_j}_{t - 1, \mathbf{n} - \mathbf{e}_j} }{f_{t - 2, \mathbf{n}}}.
	\end{equation}
	Note that this equation has a realization as an LP-algebraic object.

	Let us check that the equation satisfies all the assumptions in Section~\ref{section:introduction}.
	Let
	\[
		L = \left\{ (t, \mathbf{n}) \in \mathbb{Z} \oplus \mathbb{Z}^{a + b} \mid t + \sum^{a + b}_{i = 1} n_i \colon \text{even} \right\}
	\]
	and define $v_1, \ldots, v_{2 a + 2 b + 1} \in L$ as
	\[
		v_{2 i - 1} = (- 1, \mathbf{e}_i), \quad
		v_{2 i} = (- 1, - \mathbf{e}_i) \quad (i = 1, \ldots, a + b), \quad
		v_{2 a + 2 b + 1} = (- 2, \mathbf{0}).
	\]
	Using $h = (t, \mathbf{n})$, the equation can be written as
	\[
		f_h = \frac{ \prod^a_{i = 1} f^{k_i}_{h + v_{2 i - 1}} f^{\ell_i}_{h + v_{2 i}} + \prod^{a + b}_{j = a + 1} f^{k_j}_{h + v_{2 j - 1}} f^{\ell_j}_{h + v_{2 j}} }{f_{h + v_{2 a + 2 b + 1}}}.
	\]
	The shift vectors $v_1, \ldots, v_{2 a + 2 b + 1}$ are linearly independent over $\mathbb{Z}_{\ge 0}$ since their $t$-coordinates are all negative.
	The monoid $S$ is written as
	\[
		S = \operatorname{span}_{\mathbb{Z}_{\ge 0}}(v_1, \ldots, v_{2 a + 2 b}),
	\]
	where the number of generators in the RHS is already minimum.
	The corresponding order $\le$ is described explicitly as
	\[
		(t, \mathbf{n}) \le (t', \mathbf{n}')
		\quad \Leftrightarrow \quad
		t \le t' \text{ and }
		\sum^{a + b}_{i = 1} |n'_i - n_i| \le t' - t.
	\]
	With respect to this order, the shift vector $v_{2 a + 2 b + 1}$ is the minimum element of the set $\{ 0, v_1, \ldots, v_{2 a + 2 b + 1} \}$ because of the relation
	\[
		v_{2 i - 1} + v_{2 i} = v_{2 a + 2 b + 1}
	\]
	for $i = 1, \ldots, a + b$.
	Solving the equation in the opposite direction, we obtain
	\[
		f_{h + v_{2 a + 2 b + 1}} = \frac{ \prod^a_{i = 1} f^{k_i}_{h + v_{2 i - 1}} f^{\ell_i}_{h + v_{2 i}} + \prod^{a + b}_{j = a + 1} f^{k_j}_{h + v_{2 j - 1}} f^{\ell_j}_{h + v_{2 j}} }{f_h}.
	\]
	Therefore, the equation satisfies all the assumptions in Section~\ref{section:introduction}.
\end{example}

\begin{example}[\cite{2dchaos, exkdv, lattice_degree}]\label{example:higher}
	For an even positive integer $k$, the lattice equation
	\[
		x_{m, n} = - x_{m - 1, n - 1} + \frac{a}{x^k_{m - 1, n}} + \frac{b}{x^k_{m, n - 1}}
	\]
	is not integrable in the sense of degree growth but passes the singularity confinement test.
	One of its generalizations to a higher-dimensional lattice is
	\[
		x_{t, \mathbf{n}} + x_{t - 2, \mathbf{n}} = \sum^d_{i = 1} \left( \frac{a_i}{x^{k_i}_{t - 1, \mathbf{n} + \mathbf{e}_i}} + \frac{b_i}{x^{\ell_i}_{t - 1, \mathbf{n} - \mathbf{e}_i}} \right),
	\]
	where
	\begin{itemize}
		\item 
		$d$ is a positive integer,

		\item 
		$\mathbf{n} \in \mathbb{Z}^d$,

		\item 
		$k_i, \ell_i$ are positive even integers satisfying
		\[
			\min_{1 \le i \le d}(\ell_i k_i - 1) > \max_{1 \le  i \le d}(\ell_i, k_i),
		\]

		\item
		$a_1, b_1, \ldots, a_d, b_d$
		are variables (i.e., algebraically independent over $\mathbb{Q}$),

		\item
		the coefficient ring is $R = \mathbb{Z}[a_1, b_1, \ldots, a_d, b_d]$.

	\end{itemize}
	This equation can be Laurentified via the transformation
	\[
		x_{t, \mathbf{n}} = \frac{f_{t, \mathbf{n}} f_{t - 2, \mathbf{n}}}{F_{t - 1, \mathbf{n}}},
	\]
	where
	\[
		F_{t, \mathbf{n}} = \prod^d_{i = 1} f^{k_i}_{t, \mathbf{n} + \mathbf{e}_i} f^{\ell_i}_{t, \mathbf{n} - \mathbf{e}_i}.
	\]
	The obtained Laurent system is
	\begin{equation}\label{equation:higher}
		f_{t, \mathbf{n}}
		= - \frac{F_{t - 1, \mathbf{n}} f_{t - 4,\mathbf{n}}}{F_{t - 3,\mathbf{n}}} + \sum^d_{i = 1} \left(
		\frac{a_i F_{t - 1,\mathbf{n}} F^{k_i}_{t - 2, \mathbf{n} + \mathbf{e}_i}}{f_{t - 2, \mathbf{n}} f^{k_i}_{t - 1, \mathbf{n} + \mathbf{e}_i} f^{k_i}_{t - 3, \mathbf{n} + \mathbf{e}_i}}
		+ \frac{b_i F_{t - 1, \mathbf{n}} F^{\ell_i}_{t - 2,\mathbf{n} - \mathbf{e}_i}}{f_{t - 2, \mathbf{n}} f^{\ell_i}_{t - 1, \mathbf{n} - \mathbf{e}_i} f^{\ell_i}_{t - 3, \mathbf{n} - \mathbf{e}_i}}
		\right).
	\end{equation}
	As far as the present author knows, equation \eqref{equation:higher} is the Laurent system with the most complicated concrete expression in the literature.
	Note that this equation does not have a cluster- or LP-algebraic expression because $f_{t, \mathbf{n}}$ and $f_{t - 4, \mathbf{n}}$ do not appear in a product form, i.e., we cannot apply the Caterpillar Lemma.

	Let us check that the equation satisfies all the assumptions in Section~\ref{section:introduction}.
	Let
	\[
		L = \left\{ (t, \mathbf{n}) \in \mathbb{Z} \oplus \mathbb{Z}^d \mid t + \sum^d_{i = 1} n_i \colon \text{even} \right\}.
	\]
	We do not explicitly write down the set of shift vectors, but those vectors are linearly independent over $\mathbb{Z}_{\ge 0}$ since their $t$-coordinates are all negative.
	The monoid $S$ can be written as
	\[
		S = \operatorname{span}_{\mathbb{Z}_{\ge 0}} \{ (- 1, \pm \mathbf{e}_i) \mid i = 1, \ldots, d \}
	\]
	and the minimum element of the set of shift vectors is $v_N = (- 4, \mathbf{0})$.
	Solving the equation in the opposite direction, we obtain
	\[
		f_{t - 4, \mathbf{n}}
		= - \frac{F_{t - 3, \mathbf{n}} f_{t, \mathbf{n}}}{F_{t - 1, \mathbf{n}}} + \sum^d_{i = 1} \left(
		\frac{a_i F_{t - 3, \mathbf{n}} F^{k_i}_{t - 2, \mathbf{n} + \mathbf{e}_i}}{f_{t - 2, \mathbf{n}} f^{k_i}_{t - 1, \mathbf{n} + \mathbf{e}_i} f^{k_i}_{t - 3, \mathbf{n} + \mathbf{e}_i}}
		+ \frac{b_i F_{t - 3, \mathbf{n}} F^{\ell_i}_{t - 2, \mathbf{n} - \mathbf{e}_i}}{f_{t - 2, \mathbf{n}} f^{\ell_i}_{t - 1, \mathbf{n} - \mathbf{e}_i} f^{\ell_i}_{t - 3, \mathbf{n} - \mathbf{e}_i}}
		\right),
	\]
	the RHS of which is an irreducible Laurent polynomial since its degree with respect to $f_h$ is $1$.
	Therefore, the equation satisfies all the assumptions in Section~\ref{section:introduction}.
\end{example}

It is sometimes much easier to show the Laurent property, the irreducibility and the coprimeness on a specific domain than on a general light-cone regular domain.
For example, the present author and collaborators proved these properties for equations \eqref{equation:toda50} and \eqref{equation:higher} on some translation-invariant domains \cite{toda50,exkdv}.
Using Theorem~\ref{theorem:main}, we immediately obtain these properties on any light-cone regular domain.
Moreover, according to Proposition~\ref{proposition:not_unit}, the iterates of these equations have no unit element except on the initial boundary.

\begin{corollary}
	Equation \eqref{equation:toda50} in Example~\ref{example:toda50} has the Laurent property, the irreducibility and the coprimeness on any light-cone regular domain.
	Moreover, the iterates do not contain any nontrivial unit, i.e., if an iterate $f_{t, \mathbf{n}}$ of the equation is a Laurent monomial, then $(t, \mathbf{n})$ must belong to the initial boundary.
\end{corollary}

\begin{corollary}
	Equation \eqref{equation:higher} in Example~\ref{example:higher} has the Laurent property, the irreducibility and the coprimeness on any light-cone regular domain.
	Moreover, the iterates do not contain any nontrivial unit.
\end{corollary}

\section{Proofs}\label{section:proof}

In this section, we prove Proposition~\ref{proposition:not_unit} and Theorem~\ref{theorem:main}.

\begin{lemma}\label{lemma:initial}
	Let $H \subset L$ be light-cone regular.
	Then, its initial boundary $H_0$ can be expressed as
	\[
		H_0 = \{ h \in H \mid h + v_N \notin H \}.
	\]
\end{lemma}

\begin{proof}
	We only show that $H_0 \subset \{ h \in H \mid h + v_N \notin H \}$ since the opposite is clear by definition.
	
	Let $h \in H_0$.
	Then,
	\[
		h + v_i \notin H
	\]
	for some $i$.
	Since $v_N \le v_i$ by Assumption~\ref{assumption:invertible}, we have
	\[
		h + v_N \le h + v_i.
	\]
	Therefore, it follows from condition~\ref{enumerate:lightcone_regular_2} in Definition~\ref{definition:lightcone_regular} that $h + v_N \notin H$.
	Hence, $h$ belongs to the RHS.
\end{proof}

\begin{lemma}\label{lemma:translation}
	Let $H \subset L$ be light-cone regular and let $h_1, \ldots, h_m \in L$.
	Then, we have
	\[
		h_1, \ldots, h_m \in H + \ell v_N \quad
		(= \{ h + \ell v_N \mid h \in H \})
	\]
	for a sufficiently large $\ell \in \mathbb{Z}_{> 0}$.
\end{lemma}

\begin{proof}
	It is sufficient to consider the case $m = 1$.
	Take an arbitrary $h \in H$.
	Since the shift vectors $v_1, \ldots, v_N$ generate the lattice $L$ by Assumption~\ref{assumption:general_equation}, there exist $a_1, \ldots, a_N \in \mathbb{Z}$ such that
	\[
		h_1 - h = a_1 v_1 + \cdots + a_N v_N.
	\]
	Using
	\[
		g \ge g + v_j \ge g + v_N
	\]
	for $g \in L$ and $j = 1, \ldots, N$, we have
	\[
		h_1 = h + \sum_i a_i v_i
		\ge h + \sum_i |a_i| v_i
		\ge h + \sum_i |a_i| v_N.
	\]
	Therefore, if $\ell \ge \sum_i |a_i|$, then we have
	\[
		h_1 \ge h + \ell v_N.
	\]
	Since $h + \ell v_N \in H + \ell v_N$ and the domain $H + \ell v_N$ is light-cone regular, it follows from condition~\ref{enumerate:lightcone_regular_2} in Definition~\ref{definition:lightcone_regular} that $h_1 \in H + \ell v_N$.
\end{proof}

\begin{lemma}\label{lemma:independent_variable}
	Let $H \subset L$ be light-cone regular.
	Let $h \in H$ and $h_0 \in H_0$.
	If $h_0 \nleq h$ (meaning that $h_0 \le h$ does not hold), then $f_h$ is independent of the initial variable $f_{h_0}$.
	That is, if $f_h$ depends on the initial variable $f_{h_0}$, then $h_0 \le h$.
\end{lemma}

\begin{proof}
	Let $h_0 \nleq h$.
	We show that $f_h$ does not depend on $f_{h_0}$ by induction on $h$ with respect to the order $\le$.

	If $h \in H_0$, then $f_h$ is an initial variable.
	Therefore, it is sufficient to consider the case $h \in H \setminus H_0$.
	Since $h_0 \nleq h + v_i$ for all $i$ by the definition of the order $\le$, $f_{h + v_1}, \ldots, f_{h + v_N}$ are all independent of $f_{h_0}$.
	Hence, $f_h$, which is determined by $f_{h + v_1}, \ldots, f_{h + v_N}$, is also independent of $f_{h_0}$.
\end{proof}

\begin{lemma}\label{lemma:coefficient_extension}
	Let $H \subset L$ be light-cone regular and assume that equation \eqref{equation:general_equation} possesses the Laurent property, the irreducibility and the coprimeness.
	Then, the equation also has these three properties even when considered over the coefficient ring $K$ (the field of fractions of $R$).
	That is, for any $h \in H \setminus H_0$, $f_h$ is an irreducible element of the ring $K_{H_0}$ and for any pair $h_1, h_2 \in H$ ($h_1 \ne h_2$), $f_{h_1}$ and $f_{h_2}$ are coprime in $K_{H_0}$.
\end{lemma}

\begin{proof}
	The Laurent property immediately follows from $R_{H_0} \subset K_{H_0}$.
	If $f_h \in R_{H_0}$ is irreducible, then $f_h$ is also irreducible as an element of $K_{H_0}$, which leads to the irreducibility over $K$.
	The coprimeness follows from Lemma~\ref{lemma:coprime_extension}.
\end{proof}

\begin{remark}
	While the Laurent property and the coprimeness are preserved under any extension of the coefficient ring, the irreducibility is not in general.
	For example, the Laurent system
	\[
		f_n = \frac{f^2_{n - 1} + 1}{f_{n - 2}}
	\]
	has the irreducibility over $\mathbb{R}$ (or $\mathbb{Z}$, $\mathbb{Q}$) but the RHS itself factorizes over $\mathbb{C}$.
	Therefore, extending the coefficient ring $R$ to an algebraically closed field from the beginning does not always work well.
\end{remark}

\begin{lemma}\label{lemma:not_monomial}
	Suppose that equation \eqref{equation:general_equation} has the Laurent property on a light-cone regular domain $H$.
	Then, for any $h \in H$, there exists $h' \in H$ such that $h \le h'$ and $f_{h'}$ is not a Laurent monomial in the initial variables.
\end{lemma}

\begin{proof}
	Seeking for a contradiction, assume that there exists $h \in H$ such that for any $h' \ge h$, $f_{h'}$ is a Laurent monomial in the initial variables.
	Let $g = h - v_N$.
	Then, $f_{g + v_1}, \ldots, f_{g + v_N}$ are all Laurent monomials.
	Moreover, $f_{g + v_1}, \ldots, f_{g + v_N}$ are algebraically independent over $K$ since they have no algebraic relation.
	Therefore, it follows from Lemma~\ref{lemma:laurent_monomial_independent} that $f_g = \Phi(f_{g + v_1}, \ldots, f_{g + v_N})$ is not a Laurent monomial, which leads to a contradiction.
\end{proof}

\begin{definition}[$d_H(h)$ \cite{exkdv}]
	For a light-cone regular domain $H \subset L$ and a point $h \in H$, we define $d_H(h) \in \mathbb{Z}_{\ge 1}$ as
	\[
		d_H(h) = \# (H \cap (h + S)) \quad
		(= \# \{ h' \in H \mid h' \le h \}),
	\]
	where $\#$ denotes the cardinality of a set.
	Note that the set $H \cap (h + S)$ is always finite since $H$ is light-cone regular.
\end{definition}

\begin{lemma}
	Let $H \subset L$ be light-cone regular and suppose that $h_1, h_2 \in H$ satisfy $h_1 \le h_2$.
	Then, we have
	\[
		d_H(h_1) \le d_H(h_2)
	\]
	and the equality holds if and only if $h_1 = h_2$.
\end{lemma}

\begin{proof}
	The statement immediately follows from the inclusion
	\[
		\{ h \in H \mid h \le h_1 \} \subset \{ h \in H \mid h \le h_2 \}.
	\]
\end{proof}

\begin{lemma}\label{lemma:mutation}
	Let $H \subset L$ be light-cone regular and let $h_0$ be a minimal element of $H$ with respect to the order $\le$.
	Then, $H' = H \setminus \{ h_0 \}$ is light-cone regular, too.
	In particular, if $h \in H$, $h_0 \le h$ and $h_0 \ne h$, then $d_{H'}(h) < d_H(h)$.
\end{lemma}

\begin{proof}
	It is sufficient to show that $H'$ is light-cone regular.
	Let $h \in H'$.

	First, we show condition~\ref{enumerate:lightcone_regular_1} in Definition~\ref{definition:lightcone_regular}.
	It follows from $H' \subset H$ that
	\[
		\{ h' \in H' \mid h' \le h \} \subset \{ h' \in H \mid h' \le h \}.
	\]
	Since the RHS is a finite set, so is the LHS.

	Next, we show condition~\ref{enumerate:lightcone_regular_2} in Definition~\ref{definition:lightcone_regular}.
	Since $h_0$ is minimal in $H$ and $h \ne h_0$, we have $h \nleq h_0$.
	Hence, we have
	\[
		\{ h' \in H' \mid h' \ge h \} = \{ h' \in H \mid h' \ge h \} \subset H'.
	\]
\end{proof}

\begin{definition}[domain mutation]\label{definition:domain_mutation}
	We will call the above procedure to obtain $H'$ from $H$ the \emph{mutation} of $H$ at $h_0$.
	We also call $H'$ the mutation of $H$ at $h_0$.
\end{definition}

The term ``mutation'' comes from the theory of cluster algebras.

\begin{lemma}
	Let $H \subset L$ be light-cone regular and let $H'$ be the mutation of $H$ at a minimal element $h_0 \in H$.
	Then
	\[
		H'_0 = ( H_0 \setminus \{ h_0 \} ) \cup \{ h_0 - v_N \}.
	\]
\end{lemma}

\begin{proof}
	The statement immediately follows from Lemma~\ref{lemma:initial}.
\end{proof}

The following lemma, which gives the converse of Lemma~\ref{lemma:mutation}, is helpful when one obtains a Laurent system from a cluster or LP algebra, although we do not use it in this paper.

\begin{lemma}
	Let $H \subset L$ be light-cone regular, $h \in H$ and let $H' = H \setminus \{ h \}$.
	If $H'$ is light-cone regular, then $h$ is minimal in $H$.
\end{lemma}

\begin{proof}
	Suppose that $h$ is not minimal in $H$ and we show that $H'$ is not light-cone regular.
	Since $h$ is not minimal, there exists $h_0 \in H$ such that $h_0 \le h$ and $h_0 \ne h$.
	Since
	\[
		h_0 \in H', \quad
		h \ge h_0, \quad
		h \notin H',
	\]
	the domain $H'$ does not satisfy condition~\ref{enumerate:lightcone_regular_2} in Definition~\ref{definition:lightcone_regular}.
\end{proof}

\begin{remark}
	The above lemma implies that when one obtains a Laurent system on a lattice from a cluster algebra as in \cite{Okubo_rims, Okubo_HM}, any mutation-eligible vertex (meaning that the mutation there produces the equation) in each step must always be minimal in the sense that the corresponding lattice point is minimal in the domain with respect to the order $\le$.
\end{remark}

\begin{example}
	Consider the discrete KdV equation (Example~\ref{example:discrete_kdv}) and define $H \subset L$ as
	\[
		H = \{ (m, n) \in L \mid m, n \ge 0 \} \setminus \{ (0, 0), (1, 0) \}.
	\]
	Then, $H$ is light-cone regular (see Figure~\ref{figure:discrete_kdv_mutation}).

	Let $h_0 = (0, 1)$, which is a minimal element of $H$.
	The mutation $H'$ of $H$ at $h_0$ is
	\[
		H' = \{ (m, n) \in L \mid m, n \ge 0 \} \setminus \{ (0, 0), (1, 0), (0, 1) \}.
	\]
	Let $h = (3, 3) \in H$.
	Then, it follows from Figure~\ref{figure:discrete_kdv_mutation} that
	\[
		d_H(h) = 14, \quad
		d_{H'}(h) = 13.
	\]
\end{example}

\begin{figure}
	\begin{center} \begin{picture}(380, 150)
		
		\multiput(20, 110)(20, 0){7}{\circle{5}}
		\multiput(20, 90)(20, 0){7}{\circle{5}}
		\multiput(20, 70)(20, 0){7}{\circle{5}}
		\multiput(20, 50)(20, 0){7}{\circle{5}}
		\multiput(20, 30)(20, 0){7}{\circle{5}}
		\multiput(60, 10)(20, 0){5}{\circle{5}}
		
		\multiput(20, 110)(20, 0){2}{\circle*{5}}
		\multiput(20, 90)(20, 0){2}{\circle*{5}}
		\multiput(20, 70)(20, 0){2}{\circle*{5}}
		\multiput(20, 50)(20, 0){2}{\circle*{5}}
		\multiput(20, 30)(20, 0){4}{\circle*{5}}
		\multiput(60, 10)(20, 0){5}{\circle*{5}}
		
		\put(0, 0){\framebox(90, 80)}
		
		\put(11, 21){$h_0$}
		\put(71, 61){$h$}
		
		\put(175, 60){$\to$}
		
		\multiput(240, 110)(20, 0){7}{\circle{5}}
		\multiput(240, 90)(20, 0){7}{\circle{5}}
		\multiput(240, 70)(20, 0){7}{\circle{5}}
		\multiput(240, 50)(20, 0){7}{\circle{5}}
		\multiput(260, 30)(20, 0){6}{\circle{5}}
		\multiput(280, 10)(20, 0){5}{\circle{5}}
		
		\multiput(240, 110)(20, 0){2}{\circle*{5}}
		\multiput(240, 90)(20, 0){2}{\circle*{5}}
		\multiput(240, 70)(20, 0){2}{\circle*{5}}
		\multiput(240, 50)(20, 0){3}{\circle*{5}}
		\multiput(260, 30)(20, 0){3}{\circle*{5}}
		\multiput(280, 10)(20, 0){5}{\circle*{5}}
		
		\put(220, 0){\framebox(90, 80)}
		
		\put(291, 61){$h$}
		
	\end{picture} \end{center}
	\caption{Domain mutation at $h_0$ in the case of the discrete KdV equation.
	The initial boundaries consist of the points marked by black disks, respectively.
	The boxed regions represent the past light cones emanating from $h$, respectively.}
	\label{figure:discrete_kdv_mutation}
\end{figure}
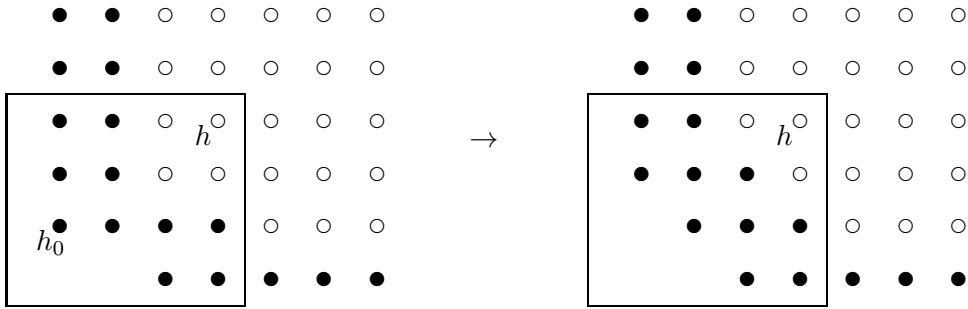


Let us prove Proposition~\ref{proposition:not_unit}.

\begin{proof}[Proof of Proposition~\ref{proposition:not_unit}]
	The ``if'' part is trivial.
	Let $h \in H \setminus H_0$.
	We show that $f_h$ is not a unit.

	Let
	\[
		B = \{ h_0 \in H_0 \mid h_0 \le h \}.
	\]
	Then, it follows from Lemma~\ref{lemma:independent_variable} that $f_h \in R_B$.
	Therefore, by Lemma~\ref{lemma:substitution}, it is sufficient to show that there exists a family $\alpha = (\alpha_b)_{b \in B}$ of elements of $\overline{K}^{\times}$ such that
	\[
		f_{h} \Big|_{f_b \leftarrow \alpha_b (b \in B)} = 0.
	\]

	By Lemma~\ref{lemma:not_monomial}, there exists $h' \in H \setminus H_0$ such that $h' \ge h$, $h' \ne h$ and $f_{h'}$ is not a Laurent monomial in the initial variables.
	Let
	\[
		\widetilde{H} = H - h' + h.
	\]
	Since $\widetilde{H}$ is nothing but a translation of $H$, the equation has the Laurent property and the coprimeness on $\widetilde{H}$.
	Moreover, $f_{h} \in R_{\widetilde{H}_0}$ is not a Laurent monomial, i.e., not a unit in this ring.
	It follows from the coprimeness on $\widetilde{H}$ that $f_h$ is coprime to $f_b$ for any $b \in B$ in the ring $R_{\widetilde{H}_0}$.
	Therefore, by Lemma~\ref{lemma:substitution}, there exists a family $(\beta_{\widetilde{h}})_{\widetilde{h} \in \widetilde{H}_0}$ of elements of $\overline{K}^{\times}$ such that
	\[
		f_{h} \Big|_{f_{\widetilde{h}} \leftarrow \beta_{\widetilde{h}} (\widetilde{h} \in \widetilde{H}_0)} = 0, \quad
		f_b \Big|_{f_{\widetilde{h}} \leftarrow \beta_{\widetilde{h}} (\widetilde{h} \in \widetilde{H}_0)} \ne 0 \quad
		(\forall b \in B).
	\]
	Let us define $\alpha_b \in \overline{K}^{\times}$ for each $b \in B$ as
	\[
		\alpha_b = f_b \Big|_{f_{\widetilde{h}} \leftarrow \beta_{\widetilde{h}} ({\widetilde{h}} \in {\widetilde{H}_0})}.
	\]
	Then, by construction, we have
	\begin{align*}
		f_h \Big|_{f_b \leftarrow \alpha_b (b \in B)}
		&= f_h \Big|_{f_{\widetilde{h}} \leftarrow \beta_{\widetilde{h}} ({\widetilde{h}} \in {\widetilde{H}_0})} = 0
	\end{align*}
	(see Figure~\ref{figure:main_strategy}).
	Hence, $f_h$ is not a unit in $R_{H_0}$.
\end{proof}

\begin{figure}
	\begin{center} \begin{picture}(400, 310)
		
		\multiput(20, 260)(20, -20){11}{\color{red} \circle*{8}}
		\multiput(20, 280)(20, -20){12}{\color{red} \circle*{8}}
		\multiput(40, 280)(20, -20){12}{\color{red} \circle*{8}}

		\multiput(60, 280)(20, -20){12}{\color{red} \circle{8}}
		\multiput(80, 280)(20, -20){12}{\color{red} \circle{8}}
		\multiput(100, 280)(20, -20){12}{\color{red} \circle{8}}
		\multiput(120, 280)(20, -20){12}{\color{red} \circle{8}}
		\multiput(140, 280)(20, -20){12}{\color{red} \circle{8}}
		\multiput(160, 280)(20, -20){12}{\color{red} \circle{8}}

		\multiput(180, 280)(20, -20){11}{\color{red} \circle{8}}
		\multiput(200, 280)(20, -20){10}{\color{red} \circle{8}}
		\multiput(220, 280)(20, -20){9}{\color{red} \circle{8}}
		\multiput(240, 280)(20, -20){8}{\color{red} \circle{8}}
		\multiput(260, 280)(20, -20){7}{\color{red} \circle{8}}
		\multiput(280, 280)(20, -20){6}{\color{red} \circle{8}}
		\multiput(300, 280)(20, -20){5}{\color{red} \circle{8}}
		\multiput(320, 280)(20, -20){4}{\color{red} \circle{8}}
		\multiput(340, 280)(20, -20){3}{\color{red} \circle{8}}
		\multiput(360, 280)(20, -20){2}{\color{red} \circle{8}}
		\multiput(380, 280)(20, -20){1}{\color{red} \circle{8}}

		\multiput(80, 280)(20, -20){12}{\circle*{4}}
		\multiput(100, 280)(20, -20){12}{\circle*{4}}
		\multiput(120, 280)(20, -20){12}{\circle*{4}}

		\multiput(140, 280)(20, -20){12}{\circle{4}}
		\multiput(160, 280)(20, -20){12}{\circle{4}}
		\multiput(180, 280)(20, -20){11}{\circle{4}}
		\multiput(200, 280)(20, -20){10}{\circle{4}}
		\multiput(220, 280)(20, -20){9}{\circle{4}}
		\multiput(240, 280)(20, -20){8}{\circle{4}}
		\multiput(260, 280)(20, -20){7}{\circle{4}}
		\multiput(280, 280)(20, -20){6}{\circle{4}}
		\multiput(300, 280)(20, -20){5}{\circle{4}}
		\multiput(320, 280)(20, -20){4}{\circle{4}}
		\multiput(340, 280)(20, -20){3}{\circle{4}}
		\multiput(360, 280)(20, -20){2}{\circle{4}}
		\multiput(380, 280)(20, -20){1}{\circle{4}}

		\put(230, 230){$h$}

		\put(140, 185){$B$}
		\multiput(150, 250)(20, -20){5}{\line(1, 0){20}}
		\multiput(110, 230)(20, -20){7}{\line(1, 0){20}}
		\multiput(170, 250)(20, -20){5}{\line(0, -1){20}}
		\multiput(110, 250)(20, -20){7}{\line(0, -1){20}}
		\put(110, 250){\line(1, 0){40}}
		\put(250, 110){\line(0, 1){40}}

		\put(10, 50){\framebox(245, 205)}
		\put(120, 55){$h + S$}

		\put(20, 290){\color{red} $\widetilde{H}_0$}
		\put(320, 40){$H_0$}

	\end{picture} \end{center}
	\caption{Main strategy of the proof of Proposition~\ref{proposition:not_unit} in the case of the discrete KdV equation.
	Here, $H$ is a staircase domain (marked with the small circles) and $\widetilde{H}$ is a translation of $H$ (marked with the big circles).
	If we consider the equation on $H$, each $f_b$ ($b \in B$) is an initial variable.
	On the other hand, if considered on $\widetilde{H}$, $f_b$ is a Laurent polynomial in $f_{\widetilde{h}}$ ($\widetilde{h} \in \widetilde{H}_0$).
	This strategy is also used in the proofs of Lemmas~\ref{lemma:main_laurent} and \ref{lemma:main_irreducibility}.}
	\label{figure:main_strategy}
\end{figure}
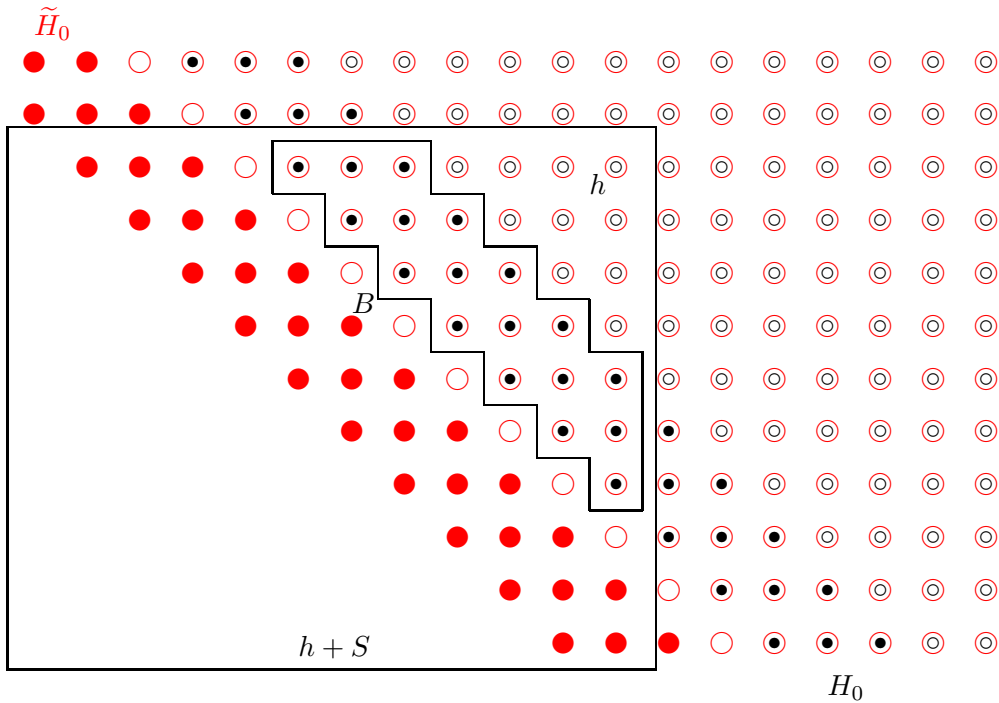

The following two lemmas play essential roles in the proof of Theorem~\ref{theorem:main}.

\begin{lemma}\label{lemma:main_laurent}
	Let $H, \widetilde{H} \subset L$ be light-cone regular.
	Suppose that equation \eqref{equation:general_equation} has the Laurent property, the irreducibility and the coprimeness on $\widetilde{H}$.
	Suppose also that the equation has the Laurent property on $H$.
	Let $h_0 \in H$ be a minimal element and let $H'$ be the mutation of $H$ at $h_0$, i.e.,
	\[
		H' = H \setminus \{ h_0 \}, \quad
		H'_0 = (H_0 \setminus \{ h_0 \}) \cup \{ h_0 - v_N \}.
	\]
	Then, the equation also possesses the Laurent property on $H'$.
\end{lemma}

\begin{proof}
	Let $h \in H'$.
	We show by induction on $d_{H'}(h)$ that $f_h \in R_{H'_0}$.
	Since the domain $H'$ is fixed throughout this proof, we simply use $d(h)$ to denote $d_{H'}(h)$.

	\newstep{step_lemma_main_laurent}
	If $h \in H'_0$, then the statement is clear.
	From here on, we only consider the case $h \in H' \setminus H'_0$.

	\newstep{step_lemma_main_laurent}
	\label{step:fhin}
	We show that $f_h \in R_{H'_0}[f^{- 1}_{h_0}]$.

	It follows from the Laurent property on $H$ that $f_h \in R_{H_0}$.
	On the other hand, Assumption~\ref{assumption:invertible} implies that $f_{h_0} \in R_{H'_0}$.
	Using the inclusion
	\[
		R_{H_0} \subset R_{H_0}[f^{- 1}_{h_0 - v_N}]
		= R_{H'_0}[f^{- 1}_{h_0}],
	\]
	we have
	\[
		f_h \in R_{H'_0}[f^{- 1}_{h_0}].
	\]

	\newstep{step_lemma_main_laurent}
	\label{step:localization}
	We show that the ring extension $R_{H'_0} \subset R_{H'_0}[f^{- 1}_{h + v_1}, \ldots, f^{- 1}_{h + v_1}]$ is a localization and that $f_h \in R_{H'_0}[f^{- 1}_{h + v_1}, \ldots, f^{- 1}_{h + v_N}]$.

	Since $d(h + v_i) < d(h)$ for each $i$, it follows from the induction hypothesis that
	\[
		f_{h + v_1}, \ldots, f_{h + v_N} \in R_{H'_0}.
	\]
	Therefore, the extension $R_{H'_0} \subset R_{H'_0}[f^{- 1}_{h + v_1}, \ldots, f^{- 1}_{h + v_N}]$ is a localization.

	Since it is defined by \eqref{equation:general_equation}, $f_h$ clearly belongs to $R_{\{ h + v_1, \ldots, h + v_N \}}$.
	Using the inclusion
	\[
		R_{\{ h + v_1, \ldots, h + v_N \}}
		\subset R_{ H'_0 \cup \{ f_{h + v_1}, \ldots, f_{h + v_N} \} }
		= R_{H'_0}[f^{- 1}_{h + v_1}, \ldots, f^{- 1}_{h + v_N}],
	\]
	we have
	\[
		f_h \in R_{H'_0}[f^{- 1}_{h + v_1}, \ldots, f^{- 1}_{h + v_N}].
	\]

	\newstep{step_lemma_main_laurent}
	By steps~\ref{step:fhin} and \ref{step:localization}, we have
	\[
		f_h \in R_{H'_0}[f^{- 1}_{h_0}] \cap R_{H'_0}[f^{- 1}_{h + v_1}, \ldots, f^{- 1}_{h + v_N}].
	\]
	If $f_{h_0}$ is coprime to $f_{h + v_i}$ in $R_{H'_0}$ for all $i$, it follows from Lemma~\ref{lemma:localization_intersection} that
	\[
		R_{H'_0} = R_{H'_0}[f^{- 1}_{h_0}] \cap R_{H'_0}[f^{- 1}_{h + v_1}, \ldots, f^{- 1}_{h + v_N}],
	\]
	which implies the statement: $f_h \in R_{H'_0}$.

	\newstep{step_lemma_main_laurent}
	From here on, we fix $i = 1, \ldots, N$ and show that $f_{h_0}$ and $f_{h + v_i}$ are coprime in $R_{H'_0}$.
	By Assumption~\ref{assumption:invertible}, $f_{h_0}$ is irreducible but not a unit in $R_{H'_0}$.
	Therefore, it is sufficient to show that $f_{h_0}$ does not divide $f_{h + v_i}$ in $R_{H'_0}$.

	\newstep{step_lemma_main_laurent}
	Since $R_{H'_0}$ has infinitely many variables, we introduce a sub-algebra with only finitely many variables.
	Let
	\[
		B = \{ b \in H'_0 \mid b \le h + v_i \text{ or } b \le h_0 - v_N \}.
	\]
	Since $H'$ is light-cone regular, $B$ is a finite set.
	By Lemma~\ref{lemma:independent_variable}, when we think of $f_{h + v_i}$ as a rational function in $f_{h'_0}$ ($h'_0 \in H_0$), the indices of the variables needed to express $f_{h + v_i}$ all belong to $B$.
	In particular, $f_{h + v_i}$ belongs to $R_B$.
	Moreover, since
	\[
		f_{h_0} \in R_{ \{ h_0 + v_1 - v_N, \ldots, h_0 + v_{N - 1} - v_N, h_0 - v_N \} }
	\]
	by Assumption~\ref{assumption:invertible}, $f_{h_0}$ also belongs to $R_B$.

	\newstep{step_lemma_main_laurent}
	We claim that there exists a family $\alpha = (\alpha_b)_{b \in B}$ of elements of $\overline{K}^{\times}$ such that
	\[
		f_{h_0} \Big|_{f_b \leftarrow \alpha_b (b \in B)} = 0, \quad
		f_{h + v_i} \Big|_{f_b \leftarrow \alpha_b (b \in B)} \ne 0.
	\]
	Assume this claim for the moment.
	Then, it follows from Lemma~\ref{lemma:substitution} that $f_{h_0}$ does not divide $f_{h + v_i}$ in $K_B$ and thus $f_{h_0}$ does not divide $f_{h + v_i}$ in the ring $R_B$, either, which completes the proof.

	Therefore, it is sufficient to show the claim, i.e., to construct such a family $\alpha = (\alpha_b)_{b \in B}$.

	\newstep{step_lemma_main_laurent}
	Let us construct such $\alpha$.

	By Lemma~\ref{lemma:translation}, we have
	\[
		B \cup \{ h_0, h + v_i \} \subset \widetilde{H} + \ell v_N
	\]
	for a sufficiently large $\ell > 0$ since the LHS is a finite set.
	Since the Laurent property, the irreducibility and the coprimeness are all invariant under a translation on the whole lattice, translating $\widetilde{H}$ if necessary, we may assume without loss of generality that
	\[
		B \cup \{ h_0, h + v_i \} \subset \widetilde{H} \setminus \widetilde{H}_0.
	\]

	By Lemma~\ref{lemma:coefficient_extension}, the Laurent property, the irreducibility and the coprimeness on $\widetilde{H}$ are all preserved under the extension of the coefficient ring from $R$ to $K$.
	Therefore, $f_{h_0}$, $f_{h + v_i}$ and $f_b$ ($b \in B$) are all irreducible elements of $K_{\widetilde{H}_0}$.
	Moreover, since $h_0 \notin H'$, $h + v_i \in H'$ and $B \subset H'$, it holds that
	\[
		h_0 \ne h + v_i, \quad
		h_0 \notin B.
	\]
	Thus, it follows from the coprimeness over $K$ that $f_{h_0}$ is coprime to any of $f_{h + v_i}$ and $f_b$ ($b \in B$) in the ring $K_{\widetilde{H}_0}$.

	It follows from Proposition~\ref{proposition:not_unit} that $f_{h_0}$ is not a unit in $K_{\widetilde{H}_0}$.
	Thus, by Lemma~\ref{lemma:substitution}, there exists a family $\beta = (\beta_{\widetilde{h}})_{\widetilde{h} \in \widetilde{H}_0}$ of elements of $\overline{K}^{\times}$ such that
	\[
		f_{h_0} \Big|_{f_{\widetilde{h}} \leftarrow \beta_{\widetilde{h}} ({\widetilde{h}} \in {\widetilde{H}_0})} = 0, \quad
		f_{h+v_i} \Big|_{f_{\widetilde{h}} \leftarrow \beta_{\widetilde{h}} ({\widetilde{h}} \in {\widetilde{H}_0})} \ne 0, \quad
		f_b \Big|_{f_{\widetilde{h}} \leftarrow \beta_{\widetilde{h}} ({\widetilde{h}} \in {\widetilde{H}_0})} \ne 0 \quad
		(\forall b \in B).
	\]
	We define $\alpha_b \in \overline{K}^{\times}$ for each $b \in B$ as
	\[
		\alpha_b = f_b \Big|_{f_{\widetilde{h}} \leftarrow \beta_{\widetilde{h}} ({\widetilde{h}} \in {\widetilde{H}_0})}.
	\]
	Then, by construction, we have
	\begin{align*}
		f_{h_0} \Big|_{f_b \leftarrow \alpha_b (b \in B)}
		&= f_{h_0} \Big|_{f_{\widetilde{h}} \leftarrow \beta_{\widetilde{h}} ({\widetilde{h}} \in {\widetilde{H}_0})} = 0, \\
		f_{h+v_i} \Big|_{f_g \leftarrow \alpha_b (b \in B)}
		&= f_{h+v_i} \Big|_{f_{\widetilde{h}} \leftarrow \beta_{\widetilde{h}} ({\widetilde{h}} \in {\widetilde{H}_0})} \ne 0,
	\end{align*}
	which implies that $f_{h_0}$ does not divide $f_{h + v_i}$.

	Hence, we have $f_h \in R_{H'_0}$ and the proof is completed.
\end{proof}


\begin{lemma}\label{lemma:main_irreducibility}
	Let $H, \widetilde{H} \subset L$ be light-cone regular.
	Suppose that equation \eqref{equation:general_equation} has the Laurent property, the irreducibility and the coprimeness on $\widetilde{H}$.
	If the equation has the Laurent property on $H$, then it also satisfies the other two properties on $H$.
\end{lemma}

\begin{proof}
	First, we prove the irreducibility.

	\newstep{step_lemma_main_irreducibility}
	Let $h \in H$.
	It follows from the Laurent property on $H$ that $f_h \in R_{H_0}$.
	We show by induction on $d_H(h)$ that $f_h$ is irreducible in $R_{H_0}$.
	If $h \in H_0$, then the statement is clear.
	Thus, we only consider the case $h \in H \setminus H_0$.

	\newstep{step_lemma_main_irreducibility}
	The proof of the irreducibility is complicated because we do not fix the domain $H$.
	To make sure, let us describe what we assume in the induction step.
	Our induction hypothesis is as follows:
	\begin{quote}
		Let $H' \subset L$ be light-cone regular and suppose that equation \eqref{equation:general_equation} has the Laurent property on $H'$.
		If $h' \in H'$ satisfies
		\[
			d_{H'}(h') < d_H(h),
		\]
		then $f_{h'}$ is irreducible in the ring $R_{H'_0}$.
	\end{quote}

	\newstep{step_lemma_main_irreducibility}
	Since the set
	\[
		\{ h' \in H \mid h' \le h \}
	\]
	is finite and nonempty, one can take a minimal element $h_0$ of this set.
	Note that $h_0 \in H_0$ and $h_0 \ne h$.
	If $h = h_0 - v_N$, then the statement is clear because $h + v_1, \ldots, h + v_N$ all belong to $H_0$.
	From here on, we only consider the case $h \ne h_0 - v_N$.

	Let $H'$ be the mutation of $H$ at $h_0$, i.e.,
	\[
		H' = H \setminus \{ h_0 \}, \quad
		H'_0 = (H_0 \setminus \{ h_0 \}) \cup \{ h_0 - v_N \}.
	\]
	By Lemma~\ref{lemma:main_laurent}, the equation has the Laurent property on $H'$ and thus we have $f_h \in R_{H'_0}$.
	Since $d_{H'}(h) < d_H(h)$, it follows from the induction hypothesis that $f_h$ is irreducible in $R_{H'_0}$.

	\newstep{step_lemma_main_irreducibility}
	It follows from Assumption~\ref{assumption:invertible} that $f_{h_0} \in R_{H'_0}$.
	Consider the inclusion relations among localized rings:
	\[
		R_{H_0} \subset R_{H_0}[f^{- 1}_{h_0 - v_N}]
		= R_{H'_0}[f^{- 1}_{h_0}] \supset R_{H'_0}.
	\]
	By Lemma~\ref{lemma:localization_irreducible}, the irreducibility of $f_h$ in $R_{H_0}$ follows from the coprimeness of $f_h$ and $f_{h_0 - v_N}$ in $R_{H_0}$.

	By Assumption~\ref{assumption:general_equation}, $f_{h_0 - v_N}$ is irreducible in $R_{H_0}$.
	Therefore, it is sufficient to show that $f_{h_0 - v_N}$ does not divide $f_h$ in $R_{H_0}$.

	\newstep{step_lemma_main_irreducibility}
	Since $R_{H_0}$ has infinitely many variables, we introduce a sub-algebra with only finitely many variables.
	Let
	\[
		B = \{ b \in H_0 \mid b \le h \text{ or } b \le h_0 - v_N \}.
	\]
	Since $H$ is light-cone regular, $B$ is a finite set.
	By Lemma~\ref{lemma:independent_variable}, when we think of $f_h$ and $f_{h_0 - v_N}$ as Laurent polynomials in $f_{h'_0}$ ($h'_0 \in H_0$), the indices of the variables needed to express $f_h$ and $f_{h_0 - v_N}$ all belong to $B$.
	In particular, we have $f_h, f_{h'_0} \in R_B$.

	\newstep{step_lemma_main_irreducibility}
	We claim that there exists a family $\alpha = (\alpha_b)_{b \in B}$ of elements of $\overline{K}^{\times}$ such that
	\[
		f_h \Big|_{f_b \leftarrow \alpha_b (b \in B)} \ne 0, \quad
		f_{h_0 - v_N} \Big|_{f_b \leftarrow \alpha_b (b \in B)} = 0.
	\]
	Assume this claim for the moment.
	By Lemma~\ref{lemma:substitution}, $f_{h_0 - v_N}$ does not divide $f_h$ in the ring $K_B$.
	Then, $f_{h_0 - v_N}$ does not divide $f_h$ in $R_B$, either, which completes the proof of the irreducibility.
	Thus, it is sufficient to construct such $\alpha$.

	\newstep{step_lemma_main_irreducibility}\label{step:construct_alpha}
	Let us construct such $\alpha$.

	Equation \eqref{equation:general_equation} has the Laurent property, the irreducibility and the coprimeness on $\widetilde{H}$.
	Since $\widetilde{H}$ is light-cone regular and $B$ is a finite set, it follows from Lemma~\ref{lemma:translation} that
	\[
		B \subset \widetilde{H} + \ell v_N
	\]
	for a sufficiently large $\ell > 0$.
	Since the Laurent property, the irreducibility and the coprimeness are all invariant under a translation on the whole lattice, translating $\widetilde{H}$ if necessary, we may assume without loss of generality that
	\[
		B \subset \widetilde{H}.
	\]

	Since $h_0 \in B \subset \widetilde{H}$, we have
	\[
		h_0 - v_N \in \widetilde{H} \setminus \widetilde{H}_0.
	\]
	It follows from Proposition~\ref{proposition:not_unit} that $f_{h_0 - v_N}$ is not a unit in $K_{\widetilde{H}_0}$.
	Thus, by Lemma~\ref{lemma:substitution}, there exists a family $\beta = (\beta_{\widetilde{h}})_{\widetilde{h} \in \widetilde{H}_0}$ of elements of $\overline{K}^{\times}$ such that
	\[
		f_{h_0 - v_N} \Big|_{f_{\widetilde{h}} \leftarrow \beta_{\widetilde{h}} ({\widetilde{h}} \in {\widetilde{H}_0})} = 0, \quad
		f_h \Big|_{f_{\widetilde{h}} \leftarrow \beta_{\widetilde{h}} ({\widetilde{h}} \in {\widetilde{H}_0})} \ne 0, \quad
		f_b \Big|_{f_{\widetilde{h}} \leftarrow \beta_{\widetilde{h}} ({\widetilde{h}} \in {\widetilde{H}_0})} \ne 0 \quad
		(\forall b \in B).
	\]

	Define $\alpha_b \in \overline{K}^{\times}$ for each $b \in B$ as
	\[
		\alpha_b = f_b \Big|_{f_{\widetilde{h}} \leftarrow \beta_{\widetilde{h}} ({\widetilde{h}} \in {\widetilde{H}_0})}.
	\]
	Then, by construction, we have
	\begin{align*}
		f_{h_0 - v_N} \Big|_{f_b \leftarrow \alpha_b (b \in B)}
		&= f_{h_0 - v_N} \Big|_{f_{\widetilde{h}} \leftarrow \beta_{\widetilde{h}} ({\widetilde{h}} \in {\widetilde{H}})} = 0, \\
		f_h \Big|_{f_b \leftarrow \alpha_b (b \in B)}
		&= f_h \Big|_{f_{\widetilde{h}} \leftarrow \beta_{\widetilde{h}} ({\widetilde{h}} \in {\widetilde{H}})} \ne 0,
	\end{align*}
	which implies that $f_{h_0 - v_N}$ does not divide $f_h$.
	Hence, we conclude that $f_h$ is irreducible in $R_{H'_0}$ and the proof of the irreducibility is finished.

	\newstep{step_lemma_main_irreducibility}
	We show the coprimeness on $H$.

	Let $h_1, h_2 \in H \setminus H_0$ and let $h_1 \ne h_2$.
	Because of the irreducibility, $f_{h_1}, f_{h_2}$ are irreducible in $R_{H_0}$.

	Define a finite set $B \subset H_0$ as
	\[
		B = \{ b \in H_0 \mid b \le h_1 \text{ or } b \le h_2 \}.
	\]
	Then, $f_{h_1}$ and $f_{h_2}$ are irreducible not only in $R_B$ but also in $K_B$.

	In order to show the coprimeness, it is sufficient to show that there exists a family $\alpha = (\alpha_b)_{b \in B}$ of elements of $\overline{K}^{\times}$ such that
	\[
		f_{h_1} \Big|_{f_b \leftarrow \alpha_b (b \in B)} \ne 0, \quad
		f_{h_2} \Big|_{f_b \leftarrow \alpha_b (b \in B)} = 0,
	\]
	since these relations imply that $f_{h_2}$ does not divide $f_{h_1}$ by Lemma~\ref{lemma:substitution}.
	However, we can construct such $\alpha$ in the same way as in step~\ref{step:construct_alpha}.

	Hence, the proof of the coprimeness is completed.
\end{proof}


Let us show our main theorem.

\begin{proof}[Proof of Theorem~\ref{theorem:main}]
	Let $H \subset L$ be light-cone regular.
	By Lemma~\ref{lemma:main_irreducibility}, it is sufficient to show the Laurent property on $H$.
	Let $h \in H \setminus H_0$.
	We show that $f_h \in R_{H_0}$.

	Roughly speaking, our strategy is to construct a sequence of domain mutations from $\widetilde{H}$ to $H$ to use Lemma~\ref{lemma:main_laurent}.
	To avoid an intricate proof, we will give a proof by induction.

	\newstep{step_theorem_main}
	Let
	\[
		G = \{ g \in H \mid g \le h \}
	\]
	and define a set of light-cone regular domains as
	\[
		\mathcal{H} = \Big\{ H' \subset L \text{: light-cone regular} \, \Big| \,
		G \subset H';
		\text{ the equation has the Laurent property on $H'$} \Big\}.
	\]
	Note that while $\mathcal{H}$ is defined only by the Laurent property, the definition implicitly includes the other two properties because of Lemma~\ref{lemma:main_irreducibility}.
	For each $H' \in \mathcal{H}$, define a nonnegative integer $D(H')$ as
	\[
		D(H') = \# \left\{ h' \in H' \mid h' \le h; \ h' \notin H \right\}.
	\]
	Roughly speaking, $D(H')$ measures the distance between the domains $H$ and $H'$ (but we ignore the region outside the past light cone emanating from $h$).

	\newstep{step_theorem_main}
	We claim that there exists $H' \in \mathcal{H}$ such that $D(H') = 0$.
	We assume this claim for the moment and show that $f_h \in R_{H_0}$.
	Let
	\[
		G' = \{ g \in H' \mid g \le h \}.
	\]

	First, we show that $G \subset G'$.
	If $g \in G$, then $g \le h$.
	Since $G \subset H'$ by the definition of $\mathcal{H}$, we have $g \in G'$.

	Next, we show that $G = G'$.
	Since $G \subset G'$, it follows from the definition of $D(H')$ that
	\[
		D(H') = \# \left( G' \setminus G \right).
	\]
	Since $D(H') = 0$, we have $G' = G$.

	Using $G' = G$, we conclude that $H$ and $H'$ coincide in the range of the past light cone emanating from $h$.
	Therefore, it follows from the Laurent property on $H'$ that $f_h \in R_{H_0}$.

	From here on, we show the claim.

	\newstep{step_theorem_main}
	We show that $\mathcal{H} \ne \emptyset$.
	Since $G$ is a finite set, by Lemma~\ref{lemma:translation}, we have
	\[
		G \subset \widetilde{H} + \ell v_N
	\]
	for a sufficiently large $\ell > 0$.
	Since the Laurent property is invariant under a translation on the whole lattice, we have
	\[
		\widetilde{H} + \ell v_N \in \mathcal{H},
	\]
	which implies $\mathcal{H} \ne \emptyset$.

	\newstep{step_theorem_main}
	Since $\mathcal{H} \ne \emptyset$, we can take $H' \in \mathcal{H}$ that achieves the minimum value of $D$.
	Let us show that $D(H') = 0$.
	Searching for contradiction, we assume that $D(H') > 0$.

	Let
	\[
		B = \left\{ h' \in H' \mid h' \le h; \ h' \notin H \right\}.
	\]
	Then, the set $B$ is finite but nonempty since $\# B = D(H') > 0$.
	Therefore, we can take a minimal element $h_0$ of $B$.

	We show that $h_0$ is minimal in $H'$.
	Suppose that $g \in H'$ satisfies $g \le h_0$.
	Since $g \le h_0$ and $h_0 \le h$, we have $g \le h$.
	On the other hand, since $g \le h_0$ and $h_0 \notin H$, it follows from the conditions on a light-cone regular domain that $g \notin H$.
	Thus, we have $g \in B$.
	Therefore, it follows from the minimality of $h_0 \in B$ that $g = h_0$.
	Hence, $h_0$ is minimal in $H'$.

	Let $H''$ be the mutation of $H'$ at $h_0$.
	By Lemma~\ref{lemma:main_laurent}, the equation has the Laurent property on $H''$.
	Moreover, since $h_0 \notin G$ by the choice of $h_0$, we have $G \subset H''$.
	Therefore, we have $H'' \in \mathcal{H}$.
	Because of $h_0$, however, $D(H'')$ must be strictly less than $D(H')$, which contradicts the minimality of $D(H')$.
	Hence, we conclude that $D(H') = 0$ and the proof is completed.
\end{proof}

\section{Reductions}\label{section:reduction}

As shown in \cite[Proposition~3.4]{investigation}, if a (possibly non-autonomous) discrete Hirota bilinear equation has the Laurent property on any light-cone regular domain, then its reductions also have the Laurent property.
Because the proof of this proposition does not rely on the specific bilinear structure of equations, it is easy to extend this result to general equations with the Laurent property on any light-cone regular domain.
Note that while it is now common, the strategy to prove the Laurent property for a lattice equation to show that an equation on a lower dimensional lattice has the Laurent property was first used in \cite{Laurent_phenomenon}, where, using the Caterpillar Lemma, the Laurent property for the Hirota-Miwa equation and the bilinear form of the discrete BKP equation was proved.

\begin{definition}[reduction. cf.\ {\cite[Definition~3.2]{investigation}}]\label{definition:reduction}
	A surjective $\mathbb{Z}$-linear map $\varphi \colon L \to L'$
	($L'$ being a lattice) is called a reduction of equation \eqref{equation:general_equation} if $\varphi(v_1), \ldots, \varphi(v_N) \in L'$ are linearly independent over $\mathbb{Z}_{\ge 0}$.
	The equation obtained by the reduction is given by
	\[
		f_{h'} = \Phi(f_{h' + \varphi(v_1)}, \ldots, f_{h' + \varphi(v_N)}) \quad
		(h' \in L'),
	\]
	which we also call a reduction of equation \eqref{equation:general_equation}.
\end{definition}

Note that we require the linear independence of $\varphi(v_1), \ldots, \varphi(v_N) \in L'$ over $\mathbb{Z}_{\ge 0}$ so that the obtained equation defines a proper evolution in the direction from $f_{h + \varphi(v_1)}, \ldots, f_{h + \varphi(v_N)}$ to $f_h$.
If $\varphi \colon L \to L'$ is a reduction, then, using $S' = \operatorname{span}_{\mathbb{Z}_{\ge 0}} \left( \varphi(v_1), \ldots, \varphi(v_N) \right)$, one can define the order $\le_{L'}$ and light-cone regular domains on $L'$ in the same way as on $L$.

\begin{proposition}[cf.\ {\cite[Proposition~3.4]{investigation}}]\label{proposition:reduction}
	Let $\varphi \colon L \to L'$ be a reduction of equation \eqref{equation:general_equation} and let $H' \subset L'$ be light-cone regular for the equation obtained by the reduction.
	Then, the lift of the domain $\varphi^{- 1}(H') \subset L$ is also light-cone regular for the original equation.
	In particular, if equation \eqref{equation:general_equation} has the Laurent property on $\varphi^{- 1}(H')$, then the equation obtained after the reduction also has the Laurent property on $H'$.
\end{proposition}

The proof we give here is based on that of \cite[Proposition~3.4]{investigation}.
The proposition in \cite{investigation} was shown under the following conditions:
\begin{itemize}
	\item
	equation \eqref{equation:general_equation} is a discrete Hirota bilinear equation,

	\item
	$\varphi$ does not decrease the number of terms of the defining equation,

	\item
	$L$ and $L'$ are free lattices.

\end{itemize}
However, we do not use these assumptions in this paper.
Since the first two assumptions were not essential even in the proof in \cite{investigation}, the only difficulty here is that lattices can now contain torsion elements.
Moreover, we will not use the following assumptions, either.
\begin{itemize}
	\item
	$\Phi$ is irreducible
	(Assumption~\ref{assumption:general_equation}),

	\item
	$\Phi$ is not a Laurent monomial
	(Assumption~\ref{assumption:general_equation}),

	\item
	There exists an irreducible Laurent polynomial $\Psi$ that solves the equation in the opposite direction
	(Assumption~\ref{assumption:invertible}).
\end{itemize}

\begin{proof}[Proof of Proposition~\ref{proposition:reduction}]
	Let $H = \varphi^{- 1}(H')$.
	We show that $H \subset L$ is light-cone regular.

	\newstep{step_proposition_reduction}
	Taking generators $x_1, \ldots, x_n \in \ker \varphi$, we can decompose $\varphi \colon L \to L'$ as
	\[
		L \to L / \mathbb{Z} x_1 \to \cdots \to L / ( \mathbb{Z} x_1 + \cdots \mathbb{Z} x_n) \cong L'.
	\]
	Therefore, it is sufficient to show that $H$ is light-cone regular in the case $L' = L / \mathbb{Z} x$ for some $x \in L$.

	\newstep{step_proposition_reduction}
	The proof in \cite{investigation} to show condition~\ref{enumerate:lightcone_regular_2} on a light-cone regular domain is still valid.
	Therefore, we only check condition~\ref{enumerate:lightcone_regular_1} in this paper.
	As the proof in \cite{investigation}, it is sufficient to show that for each $z' \in S'$, the set $\varphi^{- 1}(z') \cap S$ is finite.

	\newstep{step_proposition_reduction}
	If $x \in L$ has finite order, then the cardinality of the set $\varphi^{- 1}(z')$ coincides with the order of $x$.
	In particular, $\varphi^{- 1}(z') \cap S$ is a finite set.

	From here on, we consider the case where $x \in L$ has infinite order.

	\newstep{step_proposition_reduction}
	Let 
	\[
		L_{\mathbb{R}} = L \otimes \mathbb{R}, \quad
		L'_{\mathbb{R}} = L' \otimes \mathbb{R}, \quad
		\varphi_{\mathbb{R}} = \varphi \otimes \operatorname{id}_{\mathbb{R}} \colon L_{\mathbb{R}} \to L'_{\mathbb{R}}.
	\]
	For example, under the expression in Remark~\ref{remark:basis}, we have
	\[
		L_{\mathbb{R}} \cong \mathbb{R}^{\operatorname{rank} L}.
	\]	
	In particular, if $L$ (resp.\ $L'$) is not free, then the $\mathbb{Z}$-linear map $L \to L_{\mathbb{R}}$ (resp.\ $L' \to L'_{\mathbb{R}}$) is not injective.
	Note that this $\mathbb{Z}$-linear map decomposes into a surjection and an injection as
	\[
		L \twoheadrightarrow L / \operatorname{tor} (L) \hookrightarrow L_{\mathbb{R}},
	\]
	where $\operatorname{tor} (L)$ denotes the submodule consisting of the torsion elements of $L$.
	For $y \in L$ (resp.\ $y' \in L'$), we use $\overline{y} \in L_{\mathbb{R}}$ (resp.\ $\overline{y'} \in L'_{\mathbb{R}}$) to denote the image of $y$ by $L \to L_{\mathbb{R}}$ (resp.\ of $y'$ by $L' \to L'_{\mathbb{R}}$).
	Let
	\[
		S_{\mathbb{R}} = \operatorname{span}_{\mathbb{R}_{\ge 0}} (\overline{v}_1, \ldots, \overline{v}_N) \subset L_{\mathbb{R}}.
	\]
	Then, $S_{\mathbb{R}}$ is a closed convex cone.

	\newstep{step_proposition_reduction}\label{step:strongly_convex}
	Let us show that $\overline{\varphi(v_1)}, \ldots, \overline{\varphi(v_N)} \in L'_{\mathbb{R}}$ are linearly independent over $\mathbb{R}_{\ge 0}$.

	Using $L' / \operatorname{tor} (L') \hookrightarrow L_{\mathbb{R}}$, we can think of $\overline{\varphi(v_1)}, \ldots, \overline{\varphi(v_N)} \in L' / \operatorname{tor} (L')$ as lattice points in $L'_{\mathbb{R}}$.
	Therefore, it is sufficient to show that $\overline{\varphi(v_1)}, \ldots, \overline{\varphi(v_N)} \in L'_{\mathbb{R}}$ are linearly independent over $\mathbb{Z}_{\ge 0}$.

	Suppose that $a_1, \ldots, a_N \in \mathbb{Z}_{\ge 0}$ satisfy
	\[
		\sum_i a_i \overline{\varphi(v_i)} = 0.
	\]
	By definition, there exists $y' \in \operatorname{tor} (L')$ such that
	\[
		\sum_i a_i \varphi(v_i) = y'.
	\]
	Since $y' \in \operatorname{tor} (L')$, there exists $m \in \mathbb{Z}_{> 0}$ such that $m y' = 0$ in $L'$.
	Therefore, we have
	\[
		\sum_i m a_i \varphi(v_i) = 0,
	\]
	which implies
	\[
		a_1 = \cdots = a_N = 0
	\]
	since $\varphi(v_1), \ldots, \varphi(v_N)$ are linearly independent over $\mathbb{Z}_{\ge 0}$.

	\newstep{step_proposition_reduction}
	We show that $\varphi^{- 1}(z') \cap S$ is a finite set.
	Searching for contradiction, assume that $\varphi^{- 1}(z') \cap S$ is an infinite set.

	Let $z \in \varphi^{- 1}(z') \cap S$.
	Since $\varphi^{- 1}(z') = z + \mathbb{Z} x$ and $\varphi^{- 1}(z') \cap S$ is an infinite set, there exists an infinite sequence of distinct integers $m_1, m_2, \cdots$ such that $z + m_i x \in S$ for all $i$.
	Since
	\[
		0, \overline{z}, \overline{z} + m_i \overline{x} \in S_{\mathbb{R}}
	\]
	for all $i$, one can show that $\overline{x} \in S_{\mathbb{R}}$ in the same way as in the proof in \cite{investigation}.
	Thus, there exist $a_1, \ldots, a_N \in \mathbb{R}_{\ge 0}$ such that
	\[
		\overline{x} = \sum_i a_i \overline{v}_i.
	\]
	Since $x \in L$ has infinite order, $\overline{x} \in L_{\mathbb{R}}$ is not $0$.
	Thus, at least one of $a_1, \ldots, a_N$ must be nonzero.
	Applying $\varphi_{\mathbb{R}}$, we have
	\[
		0 = \sum_i a_i \varphi_{\mathbb{R}} (\overline{v}_i)
		= \sum_i a_i \overline{\varphi (v_i)},
	\]
	which contradicts step~\ref{step:strongly_convex}.
	Hence, $\varphi^{- 1}(z') \cap S$ is a finite set and the proof is completed.
\end{proof}

\begin{corollary}
	Any reduction of equation \eqref{equation:toda50} or \eqref{equation:higher} has the Laurent property on any light-cone regular domain.
\end{corollary}

\begin{example}
	Let us check that a reduction of equation \eqref{equation:toda50} to a one-dimensional free lattice has the following form:
	\[
		f_m = \frac{ \prod^a_{i = 1} f^{k_i}_{m - d_i} f^{\ell_i}_{m - c + d_i} + \prod^{a + b}_{j = a + 1} f^{k_j}_{m - d_j} f^{\ell_j}_{m - c + d_j} }{f_{m - c}},
	\]
	where $c, d_1, \ldots, d_{a + b} \in \mathbb{Z}$ satisfy
	\[
		0 < d_i < c, \quad
		0 < c - d_i < c, \quad
		\operatorname{GCD}(c, d_1, \ldots, d_{a + b}) = 1.
	\]
	Let $\varphi \colon L \to \mathbb{Z}$ be a reduction of \eqref{equation:toda50} and let
	\[
		c = - \varphi(v_{2 a + 2 b + 1}) = \varphi(2, \mathbf{0}), \quad
		d_i = -\varphi(v_{2 i - 1}) = \varphi(1, - \mathbf{e}_i).
	\]
	As the sign of the lattice $\mathbb{Z}$ can be changed at will, we may assume that $c > 0$.
	Then, a direct calculation shows that
	\[
		\varphi(1, \mathbf{e}_i) = - c + d_i.
	\]
	Since $\varphi(v_1), \ldots, \varphi(v_N)$ must be linearly independent over $\mathbb{Z}_{\ge 0}$, we have
	\[
		\varphi(v_i) < 0 \quad (i = 1, \ldots, 2 a + 2 b + 1),
	\]
	i.e.,
	\[
		0 < d_i < c, \quad 0 < c - d_i < c.
	\]
	The GCD condition follows from the surjectivity of $\varphi$.
\end{example}

\begin{remark}
	In \cite{investigation}, we only consider reductions that do not reduce the number of terms of a defining bilinear equation.
	The Laurent property is, however, always preserved under any reduction in the sense of Definition~\ref{definition:reduction}, even if the reduction decreases the number of terms and the $\Phi(f_{h' + \varphi(v_1)}, \ldots, f_{h' + \varphi(v_N)})$ in Definition~\ref{definition:reduction} is not irreducible as a Laurent polynomial in $f_{h' + \varphi(v_1)}, \ldots, f_{h' + \varphi(v_N)}$ anymore.
	While the Laurent property is preserved under reductions, the irreducibility and the coprimeness are not in general, as seen in Example~\ref{example:hirota_miwa_not_coprime}.
\end{remark}

\begin{remark}
	In this paper, we only considered autonomous equations.
	However, Proposition~\ref{proposition:reduction} holds even in the non-autonomous case because we only proved that the lift of a light-cone regular domain is also light-cone regular.
	To consider a reduction for a non-autonomous lattice equation, we need to impose some conditions on the equation.
	For example, if $\varphi \colon L \to L'$ is a reduction of a non-autonomous version $\Phi_h$ ($h \in L$) of equation \eqref{equation:general_equation}, then $\Phi_h$ must be $(\ker \varphi)$-invariant, i.e., $\Phi_{h + x} = \Phi_h$ for all $x \in \ker \varphi$.
\end{remark}

\section{Conclusion}

In this paper, we studied how the choice of domain for a lattice equation affects the Laurent property, the irreducibility and the coprimeness.
We showed that if it has these three properties on one light-cone regular domain, a lattice equation must also satisfy them on any light-cone regular domain.
That is, as long as we work on light-cone regular domains, these properties are independent of the choice of domain and are inherent to a lattice equation.

Our key idea was to show that an iterate of a general equation does not divide another iterate by constructing a family of elements of $\overline{K}^{\times}$ such that only the former iterate vanishes when substituting the values for the initial variables.
This strategy is a theoretical generalization of substituting concrete numerical values for initial variables, which is often used for a concrete equation.

It is sometimes much easier to show these properties on a specific domain, such as a translation-invariant one.
In fact, until now, these properties of some lattice equations were proved only on specific domains \cite{extoda, toda50, exkdv}.
Now, our main theorem guarantees that such an equation has these properties on any light-cone regular domain.

We also considered the reductions of Laurent systems.
Generalizing Proposition~3.4 in \cite{investigation}, we showed that any reduction of a lattice equation that satisfies the Laurent property on any light-cone regular domain must also satisfy the Laurent property, even if the lattices have torsion elements.

Showing the Laurent property for an ordinary difference equation is sometimes more difficult than for a lattice equation because an ordinary difference equation has only finitely many initial variables.
We sometimes take a lattice equation with the ordinary difference equation as a reduction and show that the lattice equation has the Laurent property on the domain corresponding to the reduction.
Our study in this paper allows us to replace the latter step with proving the Laurent property, the irreducibility and the coprimeness on one arbitrary light-cone regular domain.

\section*{Acknowledgement}

I wish to thank Prof.\ M. Kanki and Prof.\ T. Tokihiro for their discussion and comments.
This work was supported by a Grant-in-Aid for Scientific Research of Japan Society for the Promotion of Science, JSPS KAKENHI Grant Number 18K13438 and 23K12996.

\appendix
\section{Algebraic tools used in this study}\label{appendix:algebra}

In this appendix, we recall some basic knowledge of algebra.
We only give proofs of some lemmas since the others are well-known facts or can be shown in the usual way.
We use $X_1, X_2, Y_1, \cdots$ to denote variables of Laurent polynomial rings.
Note that in this paper, we consider units of a ring to be irreducible (Notation~\ref{notation_lattice}).

\begin{definition}\label{definition:primitive}
	Let $R$ be a UFD.
	A nonzero Laurent polynomial $f \in R \left[X^{\pm 1}_1, X^{\pm 1}_2, \cdots \right]$ is primitive if the GCD of all the coefficients of $f$ is $1$.
	That is, $f$ is primitive if and only if $f$ cannot be divided by any element of $R \setminus R^{\times}$.
\end{definition}

\begin{lemma}
	Let $R$ be a UFD and $K$ be its field of fractions.
	\begin{enumerate}
		\item
		Let $f, g \in R \left[ X^{\pm 1}_1, X^{\pm 1}_2, \cdots \right]$ and suppose that $f$ is primitive.
		If $f$ divides $g$ in the ring $K \left[ X^{\pm 1}_1, X^{\pm 1}_2, \cdots \right]$, then $f$ also divides $g$ in the ring $R \left[ X^{\pm 1}_1, X^{\pm 1}_2, \cdots \right]$.
		This is the Laurent polynomial version of Gauss's Lemma.

		\item
		Suppose that $f \in R \left[ X^{\pm 1}_1, X^{\pm 1}_2, \cdots \right]$ is not a Laurent monomial.
		Then, $f$ is irreducible in $R \left[ X^{\pm 1}_1, X^{\pm 1}_2, \cdots \right]$ if and only if $f$ is primitive in $R \left[ X^{\pm 1}_1, X^{\pm 1}_2, \cdots \right]$ and is irreducible in $K \left[ X^{\pm 1}_1, X^{\pm 1}_2, \cdots \right]$.

	\end{enumerate}
\end{lemma}

\begin{lemma}\label{lemma:coprime_extension}
	Let $R$ be a UFD and $K$ be its field of fractions.
	If $f_1, f_2 \in R \left[ X^{\pm 1}_1, X^{\pm 1}_2, \cdots \right]$ are coprime to each other, then they are also coprime as elements of $K \left[ X^{\pm 1}_1, X^{\pm 1}_2, \cdots \right]$.
\end{lemma}

\begin{proof}
	We show the contrapositive.
	Suppose that $f_1$ and $f_2$ share a nontrivial factor $g \in K \left[ X^{\pm 1}_1, X^{\pm 1}_2, \cdots \right]$.
	Multiplying an element of $K^{\times}$ if necessary, we may assume that $g$ belongs to and is primitive in $R \left[ X^{\pm 1}_1, X^{\pm 1}_2, \cdots \right]$.
	Then, it follows from Gauss's lemma that $g$ divides both $f_1$ and $f_2$ in $R \left[ X^{\pm 1}_1, X^{\pm 1}_2, \cdots \right]$.
	Hence, $f_1$ and $f_2$ are not coprime in $R \left[ X^{\pm 1}_1, X^{\pm 1}_2, \cdots \right]$.
\end{proof}

\begin{definition}
	Let $A$ be an integral domain and let $K$ be its field of fractions.
	For $f_1, \ldots, f_m \in A \setminus \{ 0 \}$, the localized ring (localization) $A \left[ f^{- 1}_1, \cdots, f^{- 1}_m \right]$ is defined as a sub-ring of $K$ as
	\[
		A \left[ f^{- 1}_1, \ldots, f^{- 1}_m \right]
		= \left\{ \frac{g}{f^{r_1}_1 \cdots f^{r_m}_m} \in K \mid g \in A; \ r_1, \ldots, r_m \in \mathbb{Z}_{\ge 0} \right\}.
	\]
	The localized ring $A \left[ f^{- 1}_1, f^{- 1}_2, \cdots \right]$ for $f_1, f_2, \cdots \in A \setminus \{ 0 \}$ is defined in the same way:
	\[
		A \left[ f^{- 1}_1, f^{- 1}_2, \cdots \right]
		= \left\{ \frac{g}{f^{r_1}_{i_1} \cdots f^{r_m}_{i_m}} \in K \mid g \in A; \ i_1, \ldots, i_m \in \mathbb{Z}_{\ge 1}; \ r_1, \ldots, r_m \in \mathbb{Z}_{\ge 0} \right\}.
	\]
\end{definition}

\begin{lemma}\label{lemma:localization_irreducible}
	Let $A$ be a UFD and let $f \in A \setminus \{ 0 \}$.
	\begin{enumerate}
		\item
		The localization $A \left[ f^{- 1} \right]$ is also a UFD.

		\item
		Let $\{ q_1, \ldots, q_m \}$ be the set of the prime elements (up to unit multiplication) that divide $f$.
		If $g \in A$ is irreducible in the localized ring $A \left[ f^{- 1} \right]$, then there exist $r_1, \ldots, r_m \in \mathbb{Z}_{\ge 0}$ and an irreducible element $p \in A$ such that
		\[
			g = p \prod^m_{i = 1} q^{r_i}_i.
		\]
		In particular, if $g \in A$ is coprime with $f$ and is irreducible in $A \left[ f^{- 1} \right]$, then $g$ is also irreducible in $A$.

	\end{enumerate}
\end{lemma}

\begin{lemma}
	Let $A$ be a UFD (for example, $A = R \left[ X^{\pm 1}_1, X^{\pm 1}_2, \cdots \right]$) and let $f \in A$.
	Then, the following conditions are equivalent:
	\begin{enumerate}
		\item
		$f$ is irreducible in $A$.

		\item
		$f$ is irreducible in the Laurent polynomial ring $A \left[ Y^{\pm 1}_1 \right]$, where $Y_1$ is an independent variable (i.e., $Y_1$ is transcendental over $A$).

	\end{enumerate}
\end{lemma}

\begin{lemma}\label{lemma:localization_intersection}
	Let $A$ be a UFD and let $f, f_1, \ldots, f_m \in A \setminus \{ 0 \}$.
	If $f$ is coprime with $f_i$ in $A$ for $i = 1, \ldots, m$, then $A$ is reproduced by localized rings as
	\[
		A = A \left[ f^{- 1} \right] \cap A \left[ f^{- 1}_1, \ldots, f^{- 1}_m \right].
	\]
\end{lemma}

\begin{lemma}\label{lemma:substitution}
	Let $K$ be a field and let $\overline{K}$ be its algebraic closure.
	Let us use $X = (X_1, \ldots, X_m)$ to denote variables and consider the Laurent polynomial ring $A = K \left[ X^{\pm 1} \right] = K \left[ X^{\pm 1}_1, \ldots, X^{\pm 1}_m \right]$.
	\begin{enumerate}
		\item\label{enumerate:substitution1}
		For $f \in A$, the following three conditions are equivalent.
		\begin{itemize}
			\item
			$f$ is a nonzero Laurent monomial in $X$.

			\item
			$f$ is a unit element of $A$.

			\item
			There does not exist a family $\alpha \in \left( \overline{K}^{\times} \right)^m$ such that
			\[
				f \Big|_{X \leftarrow \alpha} = 0,
			\]
			where we used ``$\Big|_{X \leftarrow \alpha}$'' to denote the substitution of $\alpha$ for $X$.

		\end{itemize}

		\item
		Let $f_1, f_2 \in A$.
		If there exists a family $\alpha \in \left( \overline{K}^{\times} \right)^m$ such that
		\[
			f_1 \Big|_{X \leftarrow \alpha} = 0, \quad
			f_2 \Big|_{X \leftarrow \alpha} \ne 0,
		\]
		then $f_1$ does not divide $f_2$ in $A$.

		\item\label{enumerate:substitution2}
		Let $f, f_1, \ldots, f_n \in A$ and suppose that $f$ is not a unit, i.e., $f$ is not a Laurent monomial.
		If $f$ is coprime with any of $f_1, \ldots, f_n$ in $A$, then there exists a family $\alpha \in \left( \overline{K}^{\times} \right)^m$ such that
		\[
			f \Big|_{X \leftarrow \alpha} = 0, \quad
			f_i \Big|_{X \leftarrow \alpha} \ne 0 \quad (i = 1, \ldots, n).
		\]

		\item
		All the above statements hold even when $A$ has infinitely many variables (i.e., $m = \infty$) since finitely many Laurent polynomials can contain only finitely many variables.

	\end{enumerate}
\end{lemma}

\begin{remark}
	The statements \ref{enumerate:substitution1} and \ref{enumerate:substitution2} do not hold if we only consider a family of elements of $K^{\times}$, instead of $\overline{K}^{\times}$.
	For example, $f = X^2_1 + 1 \in \mathbb{R}[X^{\pm 1}_1]$ is nonzero for the substitution of any real value for $X_1$.
\end{remark}

\begin{lemma}\label{lemma:algebraic_independence_monomial}
	Let $K \subset E$ be a field extension and let $f_1, \ldots, f_m \in E \setminus \{ 0 \}$ be distinct elements.
	Then, the following three conditions are equivalent:
	\begin{enumerate}
		\item
		$f_1, \ldots, f_m$ are algebraically independent over $K$.

		\item
		The set of monomials
		\[
			\{ f^{i_1}_1 \cdots f^{i_m}_m \mid i_1, \ldots, i_m \in \mathbb{Z}_{\ge 0} \}
		\]
		is linearly independent over $K$.

		\item
		The set of Laurent monomials
		\[
			\{ f^{i_1}_1 \cdots f^{i_m}_m \mid i_1, \ldots, i_m \in \mathbb{Z} \}
		\]
		is linearly independent over $K$.

	\end{enumerate}
\end{lemma}

\begin{lemma}\label{lemna:linear_independence_monomial}
	Let $K$ be a field and let $f_1, \ldots, f_N \in K[X^{\pm 1}_1, \ldots, X^{\pm 1}_m]$ be nonzero Laurent monomials.
	Then, $f_1, \ldots, f_N$ are linearly independent over $K$ if and only if their total degrees are all different.
	Here, the total degree of a nonzero Laurent monomial $c X^{j_1}_1 \cdots X^{j_m}_m$
	($c \in K^{\times}$, $j_1, \ldots, j_m \in \mathbb{Z}$)
	is defined as $(j_1, \ldots, j_m) \in \mathbb{Z}^m$.
\end{lemma}

Note that the ``total degree'' used in this paper differs from that introduced in \cite{domain}.

\begin{lemma}\label{lemma:laurent_monomial_independent}
	Let $K$ be a field and let $f_1, \ldots, f_N \in K[X^{\pm 1}_1, \ldots, X^{\pm 1}_m]$ be nonzero Laurent monomials.
	Let $\Phi \in K[Y^{\pm 1}_1, \ldots, Y^{\pm 1}_N]$ and consider the substitution
	\[
		\Phi\Big|_{Y_i \leftarrow f_i (\forall i)} = \Phi(f_1, \ldots, f_N) \in K[X^{\pm 1}_1, \ldots, X^{\pm 1}_m].
	\]
	If $f_1, \ldots, f_N$ are algebraically independent over $K$ and $\Phi$ is not a Laurent monomial in $Y_1, \ldots, Y_N$, then $\Phi\Big|_{Y_i \leftarrow f_i (\forall i)}$ is not a Laurent monomial in $X_1, \ldots, X_m$, either.
\end{lemma}

\begin{proof}
	We use the standard multi-index notation, such as
	\begin{itemize}
		\item 
		$X = (X_1, \ldots, X_m)$,

		\item 
		$f = (f_1, \ldots, f_N)$,

		\item
		$\Phi(f) = \Phi\Big|_{Y_i \leftarrow f_i (\forall i)}$,

		\item 
		$Y^j = Y^{j_1}_1 \cdots Y^{j_N}_N$
		for
		$j = (j_1, \ldots, j_N) \in \mathbb{Z}^N$.

	\end{itemize}
	Let us express $\Phi$ as
	\[
		\Phi = \sum_{j \in J} a_j Y^j,
	\]
	where $J \subset \mathbb{Z}^N$ is a finite set and $a_j \in K \setminus \{ 0 \}$ for $j \in J$.
	Since $f_1, \ldots, f_N$ are algebraically independent over $K$, Lemma~\ref{lemma:algebraic_independence_monomial} implies that the set $\{ f^j \mid j \in J \}$ is linearly independent over $K$.
	Therefore, by Lemma~\ref{lemna:linear_independence_monomial}, the total degrees of $f^j$ ($j \in J$) as Laurent monomials in $X$ are all different.
	Thus, the expression
	\[
		\Phi(f) = \sum_{j \in J} a_j f^j
	\]
	as an element of $K[X]$ is unique, i.e., it is impossible to decrease the number of the terms in the RHS.
	Since $\Phi$ is not a Laurent monomial in $Y$, $\# J$ must be greater than $1$.
	Hence, $\Phi(f)$ has at least $2$ terms and we conclude that it is not a Laurent monomial in $X$.
\end{proof}


\label{lastpage}
\end{document}